%% file: main.tex
\documentclass[11pt]{article}
\usepackage[left=1.00in, right=1.00in, top=1.00in, bottom=1.00in]{geometry}
\usepackage[utf8]{inputenc}
\usepackage{amsmath, amssymb, amsthm}
\usepackage{algorithm}
\usepackage[noend]{algpseudocode}
\usepackage{xcolor}
\usepackage{hyperref}
\usepackage{cleveref}
\usepackage{footnote}
\usepackage{graphicx}
\usepackage{comment}
\makesavenoteenv{algorithm}

\allowdisplaybreaks

\DeclareMathOperator*{\argmax}{arg\,max}
\DeclareMathOperator*{\argmin}{arg\,min}
\newcommand{\ceil}[1]{\lceil #1 \rceil}
\newcommand{\floor}[1]{\lfloor #1 \rfloor}

\newcommand{\E}{\mbox{\bf E}}
\newcommand{\eps}{\varepsilon}
\newcommand{\poly}{\textrm{poly}}
\newcommand{\RR}{{\mathbb R}}

\newtheorem{theorem}{Theorem}[section]
\newtheorem{definition}[theorem]{Definition}
\newtheorem{lemma}[theorem]{Lemma}
\newtheorem{observation}{Observation}[section]
\newtheorem{proposition}{Proposition}[section]

\newtheorem{remark}{Remark}[section]

\newcommand{\footref}[1]{%
    $^{\ref{#1}}$%
}

\title{Maximizing Non-Monotone Submodular Functions over Small Subsets: Beyond 1/2-Approximation}
\author{Aviad Rubinstein\thanks{Supported by NSF CCF-1954927, and a David and Lucile Packard Fellowship.}\\Stanford University\\\texttt{aviad@cs.stanford.edu} \and Junyao Zhao\thanks{Supported by NSF CCF-1954927.}\\ Stanford University\\\texttt{junyaoz@stanford.edu}}
\date{}

\begin{document}

\maketitle

\begin{abstract}
In this work we give two new algorithms that use similar techniques for (non-monotone) submodular function maximization subject to a cardinality constraint.

The first is an offline fixed parameter tractable algorithm that guarantees a $0.539$-approximation for all non-negative submodular functions. 

The second algorithm works in the random-order streaming model. It guarantees a $(1/2+c)$-approximation for  {\em symmetric} functions, and we complement it by showing that no space-efficient algorithm can beat $1/2$ for asymmetric functions. To the best of our knowledge this is the first provable separation between symmetric and asymmetric submodular function maximization.
\end{abstract}

\input{intro}

\section{Preliminaries}\label{section:prelim}
\begin{definition}
Given a ground set of elements $E$, a function $f:2^{E}\to\RR_{\ge0}$ is submodular if for all $S\subseteq T\subseteq E$ and $i\in E\setminus T$, $f(S\cup \{i\})-f(S)\ge f(T\cup \{i\})-f(T)$. Moreover, we denote the marginal gain by $f(X|S):=f(X\cup S)-f(S)$.
\end{definition}
\begin{definition}
A function $f:2^{E}\to\RR_{\ge0}$ is symmetric if for all $S\subseteq E$, $f(S)=f(E\setminus S)$.
\end{definition}
In this paper, we \textbf{always consider maximizing non-negative submodular functions} over $n$ elements under a cardinality constraint $k$, i.e., $\max_{X\in E, |X|\le k}f(X)$.
The following lemma~\cite{FMV11} for non-negative submodular functions will be useful.
\begin{lemma}
\label{lem:subsample}
Let $f\,:\,2^E \rightarrow \mathbb{R}_{\ge 0}$ be a submodular function. Further, let $R$ be a random subset of $T\subseteq E$ in which every element occurs with probability at least $p$ (not necessarily independently). Then, $\E[f(R)] \geq p f(T) + (1 - p)f(\emptyset)$.
\end{lemma}
Moreover, we are interested in the fixed-parameter tractable algorithms.
\begin{definition}
For submodular maximization over $n$ elements with cardinality constraint $k$, we say an algorithm is fixed-parameter tractable (FPT) if it has runtime $h(k)\cdot \poly(n)$, where $h$ can be any finite function. 
\end{definition}
Besides, we are also interested in studying low-memory algorithms for submodular maximization in the random-order streaming model, and in this setting, we only care about the memory cost but not the runtime. We follow the standard setup of streaming model for submodular maximization in the literature (see e.g., the model in the original work~\cite[Section 3]{BMKK14} and more recent works~\cite{AEFNS20,HKMY20,LRSVZ21}), and the only additional assumption we make is that the elements arrive in uniformly random order (which was studied in e.g.,~\cite{SMC19,LRSVZ21}). 
\subsubsection*{Random-order streaming model}\label{def:random_order_model}
In the random-order streaming model, an algorithm is given a single pass over a dataset in a streaming fashion, where the stream is a uniformly random permutation of the input dataset and each element is seen once. The algorithm is allowed to store the elements or any information in a memory buffer with certain size. To be precise, at any point during the runtime of the algorithm,
$$\textrm{memory cost}=\textrm{number of stored elements}+\textrm{number of bits of stored information}.$$
Note that the elements and information are treated separately. One can think of the elements as physical tokens\footnote{The standard streaming model for submodular maximization assumes the elements are stored like physical tokens rather than using arbitrary encoding, because the model eventually wants to restrict the algorithm's access to the value oracle. If we store elements using arbitrary encoding, it is not clear how to restrict oracle access for general submodular functions (although it is possible to define such model for some special applications). There is another model that allows elements to be stored in arbitrary encoding~\cite[Appendix B]{FNSZ20} - this model does not restrict oracle access at all, but instead it assumes that the elements appearing in the stream are a small part of the ground set.}, and the algorithm has a limited number of special slots to store the tokens. Besides these special slots, the algorithm has other limited space to store arbitrary information. The total memory cost should not exceed the algorithm's memory size.

Every time when a new element arrives, the algorithm can decide how to update its memory buffer, i.e., whether to store the new element or remove other elements in its memory, and what information to add or remove. At any time, the algorithm can make any number of queries to the value oracle of the objective submodular function, but it is only allowed to query the value of any subset of the elements that are stored in its memory. For example, at some point during the stream, suppose the algorithm stores an element $e$, and it makes a query of the value of set $\{e\}$ and writes the result of the query in its memory in any format it prefers (e.g., ``the value of $\{e\}$ is...''), and then it removes element $e$. After removing $e$, it will never be allowed to query the value of any set that contains $e$ in the future, but it can still keep the information ``the value of $\{e\}$ is...'', which it wrote before, in its memory, as long as it wants.

At the end of the stream, the algorithm outputs a subset of elements that are stored in its memory as the solution set.

\begin{remark}\label{rmk:stronger_algorithms}
Our streaming algorithm falls into the above model and has low memory cost (specifically, $\tilde{O}(k^2)$). Our hardness result in fact holds against stronger algorithms that are allowed to (i) store infinite bits of information (i.e., only the number of stored elements counts as memory cost) and (ii) output any size-$(\le k)$ subset of elements as the solution set (i.e., during the stream, the algorithm is still only allowed to query any subset of elements stored in its memory, but at the end of the stream, it can output any size-$(\le k)$ subset of the ground set as it wants\footnote{I.e., it can output something like ``My solution set is \{1,3,11,...\}'' even if elements $1,3,11$ are not in its memory.}).
\end{remark}

\section{The core algorithm}
In this section, we present the core algorithm of this work (Algorithm~\ref{alg:symmetric}), which is actually our streaming algorithm. Our FPT algorithm will use this core algorithm as a subroutine. The goal of this section is to establish the common setup for the analysis of our FPT algorithm for general submodular functions and the analysis of the streaming algorithm for symmetric submodular functions. In this process, we will also do a warm-up that proves $1/2$-approximation for the core algorithm on general submodular functions.

\begin{algorithm}
\caption{{\sc SymmetricStream}$(f, E, k, \eps)$ \label{alg:symmetric}}
\begin{algorithmic}[1]
\State Partition the first $\eps$ fraction of the random stream $E$ into windows $w_1,\dots,w_{3k}$ of equal size. 
\State $S_0\gets\emptyset$
\State $H\gets\emptyset$
\For{$i\gets1$ to $3k$}
    \State $e_i\gets\argmax_{e\in w_i} f(e|S_{i-1})$
    \State $S_i\gets S_{i-1}\cup\{e_i\}$
\EndFor
\For{$e \in E\setminus\{w_1, w_2, \ldots, w_{3k}\}$}\label{algline:main_loop}
    \For{$i=1,2,\ldots,3k$}
        \If{$f(e|S_{i-1})>f(e_i|S_{i-1})$ and $|H|< 18k^2\log k/\eps$ \label{algline:memory-threshold}}
            \State $H\gets H\cup\{e\}$
        \EndIf
    \EndFor
\EndFor
\State \textbf{return} $\argmax_{X\subseteq S_{3k}\cup H,\,|X|\le k} f(X)$
\end{algorithmic}
\end{algorithm}

At a high level, Algorithm~\ref{alg:symmetric} divides the first $\eps$ fraction of the stream into $3k$ windows\footnote{The choice of $3k$ suffices to beat $1/2$ approximation for symmetric submodular function, but it is conceivable that a larger budget might improve the final constant. For our FPT algorithm, dividing into $k$ windows would also work, but that does not improve the runtime asymptotically.} and greedily selects the element $e_i$ with the best marginal gain in each window $i$. (Although we use the notation $S_i$ in the pseudocode for clarity, we only need to keep track of one ordered solution set in the first for loop.) Then it freezes the solution set, and for the rest of the stream, it only selects the first $18k^2\log k/\eps$ elements that have better marginal gain than $e_i$ conditioned on the $(i-1)$-th partial solution for some $0\le i\le 3k-1$. Finally, it finds the best size-$k$ solution from all the selected elements by brute force. Due to line~\ref{algline:memory-threshold}, the memory usage is $O(k^2\log k/\eps)$. The runtime before the final brute-force search is $O(nk)$ because of the for loops, and the brute-force search takes time $k\binom{|S_{3k}\cup H|}{k}=2^{\widetilde{O}(k)}$ as $|S_{3k}|=3k$ and $|H|=O(k^2\log k/\eps)$, and hence, the total runtime is $O(nk)+2^{\widetilde{O}(k)}$ (which is not polynomial in $k$, but in this paper, we focus on the memory bound for streaming setting and FPT algorithms for offline setting).
\subsection{Warm-up}
As a warm-up, we prove the $1/2$-approximation of our core algorithm for general submodular functions, which also helps set up the proof of our main algorithmic results. Before getting to the technical proof, we provide the intuition for Algorithm~\ref{alg:symmetric} in the following. First, we want to make sure that with high probability $|H|$ never meets the size threshold in the if condition at line~\ref{algline:memory-threshold}, and thus the size threshold essentially does not affect our analysis. This follows by a standard argument (see Lemma~\ref{lem:whp_select_all_O_H}).

Because the stream is in random order, most optimal elements will be visited during the for loop at line~\ref{algline:main_loop}. A part of them $O_H$ will be picked by the algorithm, and the other part $O_L$ will not be selected. Consider the set $S_{|O_L|}$ defined in the algorithm. Because of the if condition at line~\ref{algline:memory-threshold}, the elements in $S_{|O_{L}|}$ have better marginal contribution than $O_{L}$, and we can show that $f(S_{|O_{L}|})\ge f(O_L|S_{|O_{L}|})$, which is somewhat similar to the classic greedy algorithm for monotone submodular maximization. Since $O_H\cup S_{|O_{L}|}$ and $S_{|O_{L}|}$ are two candidate solutions under the radar of the algorithm's final brute-force search, the algorithm achieves at least 
\begin{gather}\label{eq:warmup}\frac{f(O_H\cup S_{|O_{L}|})+f(S_{|O_{L}|})}{2}\ge\frac{f(O_H\cup S_{|O_{L}|})+f(O_L|S_{|O_{L}|})}{2}\ge\frac{f(O_H\cup O_L\cup S_{|O_{L}|})}{2},\end{gather}
where the last inequality is by submodularity.

To complete the analysis, we observe that  $f(O_H\cup O_L\cup S_{|O_L|})$ is not significantly worse than $f(O_H\cup O_L)$, and hence $1/2$-approximation follows from~\eqref{eq:warmup}. Indeed, because $S_{|O_L|}$ is chosen from a random $\eps$ fraction of $E$, it cannot hurt $O_H\cup O_L$ significantly --- otherwise, there should be many other elements similar to $S_{|O_L|}$, and together they would hurt $O_H\cup O_L$ so much that would eventually contradict non-negativity. This is formally shown in Lemma~\ref{lem:first_eps_does_not_hurt}. 
Now we start by proving Lemma~\ref{lem:first_eps_does_not_hurt} and Lemma~\ref{lem:whp_select_all_O_H} and then prove the $1/2$-approximation.

\begin{lemma}\label{lem:first_eps_does_not_hurt}
Let $O$ denote the optimal size-$k$ solution. Then, for any constants $\eps\in(0,1]$ and $\eps'>0$,
\begin{itemize}
    \item with probability at least $1-\eps/\eps'$, there does not exist a set $Y$ of elements in the first $\eps$ fraction of the stream such that $f(Y|O)\le -\eps' f(O)$,
    \item and moreover, with probability at least $1-3\eps/(\eps')^2$, for all $\ell\in\{\eps' k,2\eps' k,\dots,k\}$, for any $S\subseteq S_{3\ell}\setminus S_{2\ell}$, $f(S|O\cup S_{2\ell})\ge -\eps' f(O)$.
\end{itemize}

\end{lemma}
\begin{proof}
Divide the random stream $E$ equally into $1/\eps$ parts $E_1,\dots,E_{1/\eps}$. Let $Y_i=\argmin_{Y\subseteq E_i} f(Y|O)$. Then we have that
\begin{align*}
    \E [\min_{Y^*\subseteq E} f(Y^*|O)] &\le \E [f({\textstyle \bigcup\limits_{i\in [1/\eps]}}Y_i|O)] \\
    &\le \E [\sum_{i\in [1/\eps]}f(Y_i|O)] && \text{(By submodularity)} \\
    &= (1/\eps)\E [f(Y_1|O)]. && \text{(By linearity of expectation)}
\end{align*}
Since $f$ is non-negative, $f(Y|O)=f(Y\cup O)-f(O)\ge -f(O)$ for any $Y$. It follows that $\E [f(Y_1|O)]\ge -\eps f(O)$. By applying Markov's inequality on $-f(Y_1|O)$, the first bullet point holds with probability at least $1-\eps/\eps'$.

The second bullet point can be proved similarly. Specifically, for each $\ell$, we condition on $w_1,\dots,w_{2\ell}$ and let $E_i'$ denote the last $1-\frac{2\ell}{3k}$ fraction of $E_i$ for all $i$. Now let $Y_i'=\argmin_{Y\in E_i'} f(Y|O\cup  S_{2\ell})$ for all $i$. In particular, for any $S\subseteq S_{3\ell}\setminus S_{2\ell}\subseteq E_1'$, $f(Y_1'|O\cup  S_{2\ell})$ lower bounds $f(S|O\cup S_{2\ell})$. Then similar to above, we have that conditioned on arbitrary $w_1,\dots,w_{2\ell}$,
\begin{align*}
    \E[\min_{Y^*\subseteq E} f(Y^*|O\cup S_{2\ell})] &\le \E [f({\textstyle \bigcup\limits_{i\in [1/\eps]}}Y_i'|O\cup S_{2\ell})] \\
    &\le \E [\sum_{i\in [1/\eps]}f(Y_i'|O\cup S_{2\ell})] && \text{(By submodularity)} \\
    &= (1/\eps)\E [f(Y_1'|O\cup S_{2\ell})]. && \text{(By linearity of expectation)}
\end{align*}
By non-negativity, $f(Y|O\cup S_{2\ell})=f(Y\cup O\cup S_{2\ell})-f(O\cup S_{2\ell})\ge -f(O\cup S_{2\ell})$ for any $Y$. By optimality of $O$, $f(S_{2\ell})\le 2f(O)$, and hence by submodularity, $f(O\cup S_{2\ell})\le 3f(O)$. It follows that $\E [f(Y_1'|O\cup S_{2\ell})]\ge -3\eps f(O)$. By Markov's inequality, the second bullet point holds for each $\ell$ with probability at least $1-3\eps/\eps'$. The proof finishes by a union bound over all the $\ell$'s.
\end{proof}
\begin{lemma}\label{lem:whp_select_all_O_H}
With probability $1-3/k$, $|H|<18k^2\log k/\eps$, namely, Algorithm~\ref{alg:symmetric} never ignores any elements due to the memory threshold (line~\ref{algline:memory-threshold}).
\end{lemma}
\begin{proof}
For arbitrary $1\le i\le 3k$, conditioned on windows $w_1,\dots,w_{i-1}$, the rest of the stream starting from $w_i$ still has random order, and $w_i$ contains $\ge\eps/(3k)$ fraction of the rest of the stream starting from $w_i$. Consider the top $6k\log k/\eps$ elements with largest $f(e|S_{i-1})$ in the rest of the stream starting from $w_i$. Since for any $j$, the probability that the $j$-th top element appears in $w_i$, conditioned on that the top $j-1$ elements are not in $w_i$, is at least $\eps/(3k)$, the probability that none of the top $6k\log k/\eps$ elements appears in $w_i$ is at most $(1-\eps/(3k))^{6k\log k/\eps}\le e^{-2\log k}=k^{-2}$. Therefore, with probability at most $k^{-2}$, there are $\ge 6k\log k/\eps$ elements $e$ with $f(e|S_{i-1})>f(e_i|S_{i-1})$ in the rest of the streaming starting from $w_{i+1}$. The proof is finished by taking a union bound over all $i\in[3k]$.
\end{proof}
\begin{theorem}\label{thm:non-monotone_half_apx}
Algorithm~\ref{alg:symmetric} achieves $(1/2-O(\sqrt{\eps})-O(1/k))$-approximation for non-negative non-monotone submodular functions in the random-order streaming model.
\end{theorem}
\begin{proof}
Let $O$ be the size-$k$ optimal set, and let $O'\subseteq O$ be the subset of optimal elements that appear in the last $1-\eps$ fraction of the stream. We partition $O'$ into $O_L$ and $O_H$, where $O_L$ are the optimal elements that are not selected by the algorithm, and $O_H$ are those selected. Let $S_L$ be short for solution $S_{|O_L|}$ in the algorithm. If the good event in the first bullet point of Lemma~\ref{lem:first_eps_does_not_hurt} with $\eps'=\sqrt{\eps}$ happens (which we denote by $A_1$), then $f(S_L|O_L\cup O_H)\ge f(S_L|O)\ge-\sqrt{\eps} f(O)$, where the first inequality is by submodularity. If $A_1$ happens,
we have that
\begin{align}
    f(O_H\cup S_L)+f(O_L|S_L)&= f(O_H|S_L)+f(S_L\cup O_L) \nonumber\\
    &\ge f(O_H|S_L\cup O_L)+f(S_L\cup O_L) \quad \text{(By submodularity)} \nonumber\\
    &= f(O_L\cup O_H\cup S_L) \nonumber\\
    &= f(O_L\cup O_H)+f(S_L|O_L\cup O_H) \nonumber\\
    &\ge f(O_L\cup O_H)-\sqrt{\eps} f(O). \label{eq:asymmetric_stream_eq_1}
\end{align}
Let $O_L = \{o_1, o_2, \ldots, o_{|L|}\}$. If the good event in Lemma~\ref{lem:whp_select_all_O_H}, denoted by $A_2$, happens, then, we have the following simple observation by design of the algorithm.
\begin{observation}\label{obs:O_L_is_inferior}
If $A_2$ happens, then for any $o\in O_L$, $f(o|S_{i-1})\le f(e_i|S_{i-1})$, for all $i\in [3k]$, because $o$ is skipped by the algorithm.
\end{observation}
Given $A_2$, using the above observation, we have that
\begin{align}
    f(O_L|S_L)&\le \sum_{o\in O_L}f(o|S_L) && \text{(By submodularity)} \nonumber\\
    &\le \sum_{i=1}^{|O_L|} f(o_i|S_{i-1}) && \text{(By submodularity)} \nonumber\\
    &\le \sum_{i=1}^{|O_L|} f(e_i|S_{i-1})  && \text{(By Observation~\ref{obs:O_L_is_inferior})}\nonumber\\
    &=f(S_L). && \text{(By telescoping sum)} \label{eq:asymmetric_stream_eq_2}
\end{align}
Combining Eq.~\eqref{eq:asymmetric_stream_eq_1} and~\eqref{eq:asymmetric_stream_eq_2}, we get 
\begin{equation}
    f(O_H\cup S_L)+f(S_L)\ge f(O_L\cup O_H)-\sqrt{\eps} f(O).\label{eq:asymmetric_stream_eq_3}
\end{equation}
By Lemma~\ref{lem:subsample}, $\E [f(O')]\ge (1-\eps)f(O)$. By a Markov argument,
\begin{equation}\label{eq:asymmetric_stream_eq_4}
    \E [f(O')\mid A_1, A_2]\ge (1-\eps-\sqrt{\eps}-3/k)f(O),
\end{equation}
and hence, by taking expectation for both sides of Eq.~\eqref{eq:asymmetric_stream_eq_3}, we have that
\begin{align*}
    \E[f(O_H\cup S_L)+f(S_L)\mid A_1, A_2]&\ge \E[f(O_L\cup O_H) \mid A_1, A_2]-\sqrt{\eps} f(O) \\
    &\ge (1-\eps-2\sqrt{\eps}-3/k)f(O).
\end{align*}
Finally, $\E[f(O_H\cup S_L)+f(S_L)]\ge\Pr(A_1, A_2)\cdot\E[f(O_H\cup S_L)+f(S_L)\mid A_1, A_2]\ge (1-\sqrt{\eps}-3/k)(1-\eps-2\sqrt{\eps}-3/k)f(O)=(1-O(\sqrt{\eps}+1/k))f(O)$. The proof finishes because $O_H\cup S_L$ and $S_L$ are both subsets of $S_{3k}\cup H$, so one of them must achieve at least half of $(1-O(\sqrt{\eps}+1/k))f(O)$.
\end{proof}

\subsection{Setting up the analysis for the algorithmic results}\label{subsection:common_setup}
The proof of our main algorithmic results (Theorem~\ref{thm:symmetric_beating_half}, Theorem~\ref{thm:fpt} and Theorem~\ref{thm:fpt_plus}) is based on factor-revealing convex programs\footnote{The certificates for these convex programs can be found in appendix.}. We know that factor-revealing programs are not intuitive and hence not easy to understand, although they are effective tools for formalizing the proof. Therefore, in the future sections, before formally proving the main results, we will provide the intuition and interpretable (but less formal) analysis for our algorithmic results. \textbf{We recommend the readers only read the intuition and informal interpretable analysis in their first pass of the paper and then check the formal proofs that involve factor-revealing programs.}

In this subsection, we present the common setup of all these factor-revealing programs, which is based on a partial analysis of Algorithm~\ref{alg:symmetric}. Specifically, it contains some notations and a system of constraints that later will be a part of all our factor-revealing programs.

\subsubsection{Notations and conventions}
We define $O,O',O_L,O_H,A_2$ as in the proof of Theorem~\ref{thm:non-monotone_half_apx}. We let $A_1$ denote the union of the good events in the two bullet points of Lemma~\ref{lem:first_eps_does_not_hurt} with $\eps'=\eps^{1/3}$, which happens with probability at least $1-3\eps^{1/3}-\eps^{2/3}$. Since $\E [f(O')]\ge (1-\eps)f(O)$ by Lemma~\ref{lem:subsample} and $f(O')\le f(O)$ by optimality of $O$, it follows by a Markov argument that with probability $\le \eps^{1/2}$, $f(O')\le (1-\eps^{1/2})f(O)$. Let $A_3$ denote the event $f(O')\ge (1-\eps^{1/2})f(O)$. We \textbf{always condition on} the events $A_1,A_2,A_3$, which happen with probability at least $1-O(\eps^{1/3}+1/k)$. We also assume $f(O)=1$. Moreover, we let $\ell:=\ceil{\frac{|O_L|}{\eps^{1/3} k}}\cdot\eps^{1/3} k$, and hence, $|O_L|\le\ell\le|O_L|+\eps^{1/3} k$. We let $S_L:=S_{\ell}$ (for $S_{\ell}$ from the algorithm), and also $S_L' := S_{2\ell}\setminus S_L$ and $S_L'' := S_{3\ell}\setminus (S_L\cup S_L')$, which are well-defined because $\ell\le k$. Furthermore, we introduce four variables $a:=f(O_H|S_L)$, $b:=f(S_L)$, $c:=f(O_H)$ and $d:=f(S_L'|O_H)$, which later will be used in our factor-revealing programs. We will use $\approx$, $\lesssim$ and $\gtrsim$ to represent (in)equalities up to an additive error of $\textrm{poly}(\eps+1/k)$, i.e., $x\approx y$ means $y-\textrm{poly}(\eps+1/k)\le x\le y+\textrm{poly}(\eps+1/k)$, and $x\lesssim y$ means $x\le y+\textrm{poly}(\eps+1/k)$, and $x\gtrsim y$ means $x\ge y-\textrm{poly}(\eps+1/k)$.

\subsubsection{System of constraints for factor-revealing programs}
Our FPT algorithms will use some non-streaming sub-procedures in addition to Algorithm~\ref{alg:symmetric}, and the streaming result about symmetric submodular functions needs the assumption that the function is symmetric. However, the system of constraints in the following Lemma~\ref{lem:factor_revealing_common_constraints} is a common part of our factor-revealing programs in the proofs of both our FPT result and streaming result.
Thus, we believe this is the right place to state and prove this lemma, so that the reader can easily verify that it does not rely on non-streaming part of our FPT algorithm or the symmetric assumption of our streaming result.
\begin{lemma}\label{lem:factor_revealing_common_constraints}
Following the notations and conventions made in this subsection, Algorithm~\ref{alg:symmetric} satisfies either Eq.~\eqref{eq:symmetric_stream_1} and \eqref{eq:symmetric_stream_3}, or Eq.~\eqref{eq:symmetric_stream_1} and \eqref{eq:symmetric_stream_4}.
\end{lemma}
\begin{proof}
By event $A_1,A_3$ and submodularity, $d\gtrsim 0$ and $a+b\gtrsim c$. (As an example, one can derive $d\gtrsim 0$ by $f(S_L'|O_H)\ge f(S_L'|O)\gtrsim -\eps^{1/3} f(O)$.) By non-negativity, $b,c\ge0$.
Without loss of generality, we assume $b,c,d,a+b,c+d\le 0.9$, otherwise Algorithm~\ref{alg:symmetric} will achieve close-to-$0.9$ approximation because $S_L,S_L',O_H,O_H\cup S_L,O_H\cup S_L'$ are all candidate solutions in the final exhaustive search of the algorithm. By submodularity, $a\le c$. By Eq.~\eqref{eq:asymmetric_stream_eq_1} (the derivation still goes through for the new definition of $S_L$ in this subsection) and event $A_3$,
\begin{equation}\label{eq:symmetric_stream_0.25}
    f(O_L|S_L)\gtrsim 1-a-b,
\end{equation}
and similar to Eq.~\eqref{eq:asymmetric_stream_eq_2}, we derive that
\begin{align}
    f(S_L)&=\sum_{i=1}^{\ell} f(e_i|S_{i-1}) && \text{(By telescoping sum)}\nonumber\\
    &\ge\sum_{i=1}^{\ell} \frac{\sum_{o\in O_L}f(o|S_{i-1})}{|O_L|} && \text{(By Observation~\ref{obs:O_L_is_inferior})}\nonumber\\
    &=\sum_{i=1}^{\ell} \frac{\sum_{o\in O_L}f(o|S_L)}{|O_L|} && \text{(By submodularity)} \nonumber\\
    &\le \sum_{i=1}^{\ell} \frac{f(O_L|S_L)}{|O_L|} && \text{(By submodularity)} \nonumber\\
    &\ge f(O_L|S_L), && \text{(By $\ell\ge|O_L|$)} \label{eq:symmetric_stream_0.5}
\end{align}
which says $f(O_L|S_L)\le b$.
Combining this with Eq.~\eqref{eq:symmetric_stream_0.25}, we have that $1-a-b\lesssim b$. Summarizing all the initial constraints,
\begin{align}
     &0\lesssim b,c,d\le 0.9 \nonumber\\
     &a\le c\nonumber\\
     &a+b\gtrsim c \nonumber\\
     &a+2b\gtrsim 1 \nonumber\\
     &a+b,\,c+d\le 0.9. \label{eq:symmetric_stream_1}
\end{align}

Now we establish additional constraints. By Observation~\ref{obs:O_L_is_inferior} and submodularity, for all $i\le|O_L|$ and all $o\in O_L$, $f(e_i|S_{i-1})\ge f(o|S_{i-1})\ge f(o|S_L)$, and hence, $f(e_i|S_{i-1})\ge\frac{\sum_{o\in O_L}f(o|S_L)}{|O_L|}\ge\frac{f(O_L|S_L)}{|O_L|}$, where the last inequality is again by submodularity. Therefore, for all $0\le \gamma\le 1$,
\begin{equation}\label{eq:symmetric_stream_2}
    f(S_L)\ge f(S_{\gamma|O_L|})+\sum_{i=\gamma|O_L|+1}^{|O_L|}f(e_i|S_{i-1})\ge f(S_{\gamma|O_L|})+(1-\gamma)f(O_L|S_L).
\end{equation}
Moreover, similar to Eq.~\eqref{eq:symmetric_stream_0.5}, we can derive that $f(S_{\gamma|O_L|})\ge \gamma f(O_L|S_{\gamma|O_L|})$. By event $A_1$ and submodularity, $f(O_L|S_{\gamma|O_L|})=f(O_L\cup S_{\gamma|O_L|})-f(S_{\gamma|O_L|})\gtrsim f(O_L)-f(S_{\gamma|O_L|})$. It follows that $f(S_{\gamma|O_L|})\gtrsim \gamma(f(O_L)-f(S_{\gamma|O_L|}))$, which is equivalent to $f(S_{\gamma|O_L|})\gtrsim \gamma f(O_L)/(1+\gamma)$ by rearranging. Combining this with Eq. \eqref{eq:symmetric_stream_2}, we get for all $0\le \gamma\le 1$,
\begin{equation}\label{eq:symmetric_stream_2.5}
    f(S_L)\gtrsim \gamma f(O_L)/(1+\gamma)+(1-\gamma)f(O_L|S_L).
\end{equation}
Recall that $f(S_L)=b$ and $f(O_L|S_L)\gtrsim 1-a-b$ by Eq.~\eqref{eq:symmetric_stream_0.25}, and notice that $f(O_L)\ge f(O')-f(O_H)=(1-c)f(O')\approx 1-c$ by submodularity and event $A_3$. Plugging these inequalities into Eq.~\eqref{eq:symmetric_stream_2.5}, we get that for all $0\le \gamma\le 1$,
\begin{align*}
    b&\gtrsim \gamma (1-c)/(1+\gamma)+(1-\gamma)(1-a-b)
\end{align*}
which after rearranging is equivalent to (note that $a+b-1\neq 0$ because $a+b\le 0.9$ in Eq.~\eqref{eq:symmetric_stream_1})
$$(a+b-1)\gamma^2+\gamma(1-c-b)+1-a-2b\lesssim 0.$$
If $1-c-b\ge 0$, the left hand side is maximized by $\gamma=\frac{1-c-b}{2(1-a-b)}$, and otherwise, it is maximized by $\gamma=0$. The two cases result into two possible groups of constraints
\begin{align}
    &1-c-b\ge 0\nonumber\\
    (1-c-b)^2+&4(1-a-b)(1-a-2b)\lesssim 0, \label{eq:symmetric_stream_3}
\end{align}
where the quadratic function on the left of the second inequality turns out to be convex (one can check its Hessian matrix, but our certificates in the appendix do not rely on convexity), or
\begin{align}
    &1-c-b\le 0\nonumber\\
    &1-a-2b\lesssim 0. \label{eq:symmetric_stream_4}
\end{align}
The final system of constraints is either Eq.~\eqref{eq:symmetric_stream_1} and \eqref{eq:symmetric_stream_3}, or Eq.~\eqref{eq:symmetric_stream_1} and \eqref{eq:symmetric_stream_4}.
\end{proof}

\section{FPT algorithms for non-monotone submodular functions}
In this section, building on Algorithm~\ref{alg:symmetric}, we give FPT algorithms that achieves better-than-$1/2$ approximation for general non-monotone submodular functions. We first present a basic FPT algorithm that achieves $0.512$-approximation to show the main ideas, then we discuss how to improve the basic algorithm to get $0.539$-approximation.
\subsection{The basic FPT algorithm}
Essentially, after running Algorithm~\ref{alg:symmetric}, our basic FPT algorithm (Algorithm~\ref{alg:fpt}) searches for $O_H$ by brute force, and then starting with $O_H$ as the initial solution set, it runs classic greedy algorithm to construct a size-$k$ solution set, and it repeats this step many times without replacement, i.e., each time the elements selected by greedy algorithm are removed from ground set, and finally it returns the best size-$k$ solution set among all the repetitions. The pseudocode is given in Algorithm~\ref{alg:fpt}.

\begin{algorithm}
\caption{{\sc FPT}$(f, E, k, \eps, T)$ \label{alg:fpt}}
\begin{algorithmic}[1]
\State Initialize an empty set $X^*$.
\State Run Algorithm~\ref{alg:symmetric} on the input\footnote{Randomly permute $E$ if it is not in random order.} $(f, E, k, \eps)$ and keep $S_{3k}$ and the final version of $H$ in Algorithm~\ref{alg:symmetric}.
\For{each size-$(\le k)$ subset $O_H^{\textrm{guess}}\subseteq H$}
    \State Initialize an empty set of elements $I$.
    \For{$i=1,2,\ldots,T$}
        \State Run greedy algorithm with $O_H^{\textrm{guess}}$ as the initial solution set\footnote{Specifically, the greedy algorithm starts with solution set $X=O_H^{\textrm{guess}}$ and runs in $k-|X|$ iterations. In each iteration, it selects the element $e$ in $E$ that maximizes $f(e|X)$ and add $e$ to $X$. (\textbf{We assume without loss of generality} that the maximal $f(e|X)$ is always non-negative. Otherwise, we can add dummy elements to $E$.)} to build a size-$k$ solution set $O_H^{\textrm{guess}}\cup X_i$. Add $X_i$ to $I$ and remove $X_i$ from $E$.
    \EndFor
    \State Let $X'=\argmax_{X\subseteq S_{3k}\cup H\cup I,\,|X|\le k} f(X)$ and let $X^*=X'$ if $f(X')>f(X^*)$.\label{algline:brute_force}
    \State Add $I$ back to $E$.
\EndFor
\State \textbf{return} $X^*$
\end{algorithmic}
\end{algorithm}

Algorithm~\ref{alg:fpt} runs in fixed-parameter polynomial time. Indeed, because $|H|=\widetilde{O}(k^2)$ by design of Algorithm~\ref{alg:symmetric}, the outer loop has less than $k\binom{\widetilde{O}(k^2)}{k}=2^{\widetilde{O}(k)}$ iterations, and for the inner loop, we will only need $T$ to be an arbitrarily large constant. Greedy algorithm runs in $O(kn)$ time. Moreover, since $|I|\le Tk$ and $|S_{3k}\cup H|=\widetilde{O}(k^2)$ with high probability, the brute-force step (Line~\ref{algline:brute_force}) in Algorithm~\ref{alg:fpt} takes time $k\binom{\widetilde{O}(k^2)}{k}=2^{\widetilde{O}(k)}$. Furthermore, the runtime of Algorithm~\ref{alg:symmetric} excluding the exhaustive search in its last step is polynomial. Hence, the total runtime is $n\cdot2^{\widetilde{O}(k)}$.

Before proving the approximation ratio for Algorithm~\ref{alg:fpt}, we give an interpretable analysis for the better-than-1/2 approximation. At very high level, the intuition is that if Algorithm~\ref{alg:symmetric} only gets $1/2$ approximation, then it must be the case that $f(O_H)=\frac{f(O_H\cup O_L)}{2}$. Now we consider the candidate solution $O_H\cup X_i$ for $i\in \{1,2,3\}$, and we can argue that if $f(X_i|O_H)=0$, i.e., the candidate solution does not beat $1/2$, then $X_i$ must hurt $O_H\cup O_L$ a lot. Moreover, By submodularity, $X_1\cup X_2\cup X_3$ hurts $O_H\cup O_L$ by at least the sum of how much each $X_i$ ($i\in \{1,2,3\}$) hurts. Together, we show that this would contradict non-negativity of the function. Now we explain this intuition in more details.

\paragraph{Informal interpretable analysis} The starting point is the analysis of Theorem~\ref{thm:non-monotone_half_apx}. We can show that for the instance to be hard, in the sense that Algorithm~\ref{alg:symmetric} is only able to get $1/2$ approximation, then it requires $f(O_H)=\frac{f(O_H\cup O_L)}{2}$ (this will be explained with more details in the interpretable analysis provided before Theorem~\ref{thm:symmetric_beating_half}, but for now, let us take this as given). Because $O_H$ is selected by Algorithm~\ref{alg:symmetric}, in the outer iteration when Algorithm~\ref{alg:fpt} guesses $O_H$ correctly, it runs classic greedy algorithm many times based on $O_H$ without replacement. Consider the set $X_1$ selected in the first run of greedy algorithm.
By standard analysis of greedy algorithm, we can derive that $f(X_1|O_H)\ge f(O_L|O_H\cup X_1)$.
If $f(X_1|O_H)=0$ (otherwise $X_1\cup O_H$ beats $1/2$), then $f(O_L|O_H\cup X_1)\le 0$, which implies $f(O_L\cup O_H\cup X_1)\le f(O_H\cup X_1)=f(O_H)=\frac{f(O_H\cup O_L)}{2}$. Hence $f(X_1|O_L\cup O_H)\le-\frac{f(O_H\cup O_L)}{2}$. WLOG, the first run of greedy did not select most of $O_L$, because otherwise $f(X_1|O_H)$ should be significantly large. Therefore, similarly, we can derive that if $f(X_2|O_H)=0$, where $X_2$ is selected in the second run of greedy, then $f(X_2|O_L\cup O_H)\le-\frac{f(O_H\cup O_L)}{2}$. By submodularity, $f(X_1\cup X_2|O_L\cup O_H)\le f(X_1|O_L\cup O_H)+f(X_2|O_L\cup O_H)\le-f(O_H\cup O_L)$. Notice that this implies $f(X_1\cup X_2\cup O_L\cup O_H)\le 0$. Hence, the third run of greedy algorithm must obtain very large $f(X_3|O_H)$ (and hence $X_3\cup O_H$ beats $1/2$), because otherwise we can repeat above argument and show that $f(X_1\cup X_2\cup X_3\cup O_L\cup O_H)<0$, which violates non-negativity of the function $f$. Furthermore, by running greedy many times, we are able to extract even more value from $O_L$, which is formally formulated by the factor-revealing programs in the proof.

\begin{theorem}\label{thm:fpt}
For sufficiently large constant $T$ and sufficiently small constant $\eps$, Algorithm~\ref{alg:fpt} achieves $0.512$-approximation for non-negative non-monotone submodular maximization with a cardinality constraint.
\end{theorem}
\begin{proof}
In this proof, we \textbf{always consider} the iteration of the outer loop when $O_H^{\textrm{guess}}=O_H$, and we want to show the $X'$ in that iteration achieves $(>0.512)$-approximation.
We follow the setup given in Subsection~\ref{subsection:common_setup}.
There we established a system of constraints, which is either Eq.~\eqref{eq:symmetric_stream_1} and \eqref{eq:symmetric_stream_3}, or Eq.~\eqref{eq:symmetric_stream_1} and \eqref{eq:symmetric_stream_4}, for Algorithm~\ref{alg:symmetric}, which is a subroutine of Algorithm~\ref{alg:fpt}. This system of constraints will be a part of the final constraints of the factor-revealing programs for this proof.

\subsubsection*{System of constraints for a factor-revealing program}
Now we add some new constraints to the constraint system for the factor-revealing program in this proof. First, we let $\alpha_i:=f(X_i|O_H)$, where $X_i$ is defined in Algorithm~\ref{alg:fpt}, and let $\alpha:=\sum_{i\in[T]}\alpha_i/T$. By submodularity, $f(X_i|O_H)\le f(X_i)\le f(O)=1$, and as mentioned in the footnote in Algorithm~\ref{alg:fpt}, $f(X_i|O_H)\ge 0$. Hence,
\begin{equation}
\forall\,i\in[T],\,0\le \alpha_i\le 1,
\end{equation}
and it follows that
\begin{equation}\label{eq:fpt_new_constraint_1}
    0\le \alpha\le 1.
\end{equation}

Next, we partition $O_L$ into $O_L^{(0)}$ and $O_L^{(1)}$ such that for each $o\in O_L$, $o\in O_L^{(1)}$ iff $o\in X_i$ for some $i\in T$, namely, $O_L^{(1)}$ contains the part of $O_L$ that is selected during that outer iteration of the algorithm, and $O_L^{(0)}$ contains the part that is not selected. Let $\beta_0:=f(O_L^{(0)}|O_H)$ and $\beta_1:=f(O_L^{(1)}|O_H)$. Obviously, $0\le f(O_L^{(0)}|O_H),f(O_L^{(1)}|O_H)\le f(O)=1$ (for example, if $f(O_L^{(0)}|O_H)<0$, then by submodularity, $f(O_L^{(0)}|O \setminus O_L^{(0)})<0$, but then $O \setminus O_L^{(0)}$ is strictly better than the optimal solution $O$, and if $f(O_L^{(0)}|O_H)>f(O)$, then by submodularity, $O_L^{(0)}$ is strictly better than $O$)%\aviad{the non-negativity took me a minute. maybe good to explicitly explain it.}
, and by submodularity, $f(O_L^{(0)}|O_H)+f(O_L^{(1)}|O_H)+f(O_H)\ge f(O)$. Thus, we have
\begin{equation}\label{eq:fpt_new_constraint_2}
\beta_0+\beta_1+c\ge 1 \textrm{ and } 0\le\beta_0,\beta_1\le 1.
\end{equation}
 Suppose $X_i=\{x_1,\dots,x_{|O_L|}\}$ and $O_L^{(0)}=\{o_1,\dots,o_{|O_L^{(0)}|}\}$. Let $X_i^{(j)}:=\{x_1,\dots,x_j\}$. We derive that
\begin{align*}
    f(X_i|O_H)&=\sum_{1\le j\le |O_L|} f(x_j | O_H\cup X_i^{(j-1)}) &&\text{(Telescoping sum)}\\
    &\ge \sum_{1\le j\le |O_L^{(0)}|} f(x_j | O_H\cup X_i^{(j-1)}) &&\text{($f(x_j | O_H\cup X_i^{(j-1)})\ge 0$ for all $j$)}\\
    &\ge \sum_{1\le j\le |O_L^{(0)}|} f(o_j | O_H\cup X_i^{(j-1)}) &&\text{(By greedy selection)}\\
    &\ge \sum_{1\le j\le |O_L^{(0)}|} f(o_j | O_H\cup X_i) &&\text{(By submodularity)}\\
    &\ge f(O_L^{(0)} | O_H\cup X_i) &&\text{(By submodularity)}.
\end{align*}
Expanding both sides, we have that
\begin{align*}
    f(O_H\cup X_i)-f(O_H)&\ge f(O_L^{(0)} \cup O_H\cup X_i)-f(O_H\cup X_i)\\
    f(O_L^{(0)} \cup O_H\cup X_i)&\le 2f(O_H\cup X_i)-f(O_H) &&\text{(Rearranging)}\\
    f(X_i|O_L^{(0)} \cup O_H)+f(O_L^{(0)} \cup O_H)&\le 2f(X_i|O_H)+f(O_H) &&\text{(Expanding both sides)},
\end{align*}
and by rearranging,
\begin{align}\label{eq:fpt_eq_1}
    f(X_i|O_L^{(0)} \cup O_H)&\le 2f(X_i|O_H)-f(O_L^{(0)}|O_H)\nonumber\\
    &= 2\alpha_i-\beta_0.
\end{align}
Furthermore, notice that
\begin{align}\label{eq:fpt_eq_2}
    \sum_{i\in[T]} f(X_i|O_L^{(0)} \cup O_H)&\ge f(\bigcup_{i\in[T]} X_i|O_L^{(0)} \cup O_H) &&\text{(By submodularity)}\nonumber\\
    &\ge -f(O_L^{(0)} \cup O_H) &&\text{(By non-negativity)}\nonumber\\
    &=-f(O_L^{(0)}| O_H)-f(O_H)=-\beta_0-c.
\end{align}
Combining Eq.~\eqref{eq:fpt_eq_1} and Eq.~\eqref{eq:fpt_eq_2}, we get $2\sum_{i\in[T]}\alpha_i/T-\beta_0\ge(-\beta_0-c)/T$. Since $\beta_0,c\in[0,1]$ and recall $\alpha:=\sum_{i\in[T]}\alpha_i/T$, for sufficiently large constant $T$,
\begin{equation}\label{eq:fpt_new_constraint_3}
    2\alpha\gtrsim\beta_0.
\end{equation}
Our final system of constraints is either Eq.~\eqref{eq:symmetric_stream_1}~\eqref{eq:symmetric_stream_3}~\eqref{eq:fpt_new_constraint_1}~\eqref{eq:fpt_new_constraint_2}~\eqref{eq:fpt_new_constraint_3} or Eq.~\eqref{eq:symmetric_stream_1}~\eqref{eq:symmetric_stream_4}~\eqref{eq:fpt_new_constraint_1}~\eqref{eq:fpt_new_constraint_2}~\eqref{eq:fpt_new_constraint_3}.

\subsubsection*{Objective for factor-revealing program}
We consider the following candidate solutions in the final enumeration of the algorithm --- $S_L$, $S_L\cup O_H$, $O_L^{(1)}\cup O_H$, $X_i\cup O_H$ for all $i\in[T]$. Since $f(S_L)=b$, $f(S_L\cup O_H)=a+b$, $f(O_L^{(1)}\cup O_H)=\beta_1+c$, and
$$\frac{1}{T}\cdot\sum_{i\in [T]} f(X_i\cup O_H)=\left(\frac{1}{T}\cdot\sum_{i\in [T]} f(X_i|O_H)\right) + f(O_H)=\alpha+c,$$
%\aviad{the last one is less than the max of this average which we consider anyway, no?}
the algorithm achieves at least
\begin{equation*}
    r_{\textrm{FPT}}(a,b,c,\alpha,\beta_0,\beta_1):=\max\{b,\,a+b,\,\beta_1+c,\,\alpha+c\}.
\end{equation*}
We can find the lower bound of approximation ratio by numerically solving the following two convex programs,
$$\min_{a,b,c,\alpha,\beta_0,\beta_1} r_{\textrm{FPT}}(a,b,c,\alpha,\beta_0,\beta_1) \textrm{ s.t. Eq.~\eqref{eq:symmetric_stream_1}~\eqref{eq:symmetric_stream_3}~\eqref{eq:fpt_new_constraint_1}~\eqref{eq:fpt_new_constraint_2}~\eqref{eq:fpt_new_constraint_3}},$$
$$\min_{a,b,c,\alpha,\beta_0,\beta_1} r_{\textrm{FPT}}(a,b,c,\alpha,\beta_0,\beta_1) \textrm{ s.t. Eq.~\eqref{eq:symmetric_stream_1}~\eqref{eq:symmetric_stream_4}~\eqref{eq:fpt_new_constraint_1}~\eqref{eq:fpt_new_constraint_2}~\eqref{eq:fpt_new_constraint_3}},$$
and the minimum of the two results turns out to be larger than $0.512$.
\end{proof}

\subsection{Improved FPT algorithm}
Algorithm~\ref{alg:fpt} can be improved to achieve better approximation ratio. In this subsection, we present the improved algorithm, the pseudocode of which is given in Algorithm~\ref{alg:fpt_plus}. The observation for the improvement is that in the proof of Theorem~\ref{thm:fpt}, we get the value of $O_L^{(1)}$ and $O_L^{(0)}$ separately (specifically, we get the value of $O_L^{(1)}$ because of the candidate solution $O_L^{(1)}\cup O_H$, and we get some fraction of the value of $O_L^{(0)}$ indirectly through the candidate solutions $X_i\cup O_H$), and therefore, the worst case is at some balance point where both $O_L^{(1)}$ and $O_L^{(0)}$ have certain non-negligible value. In Algorithm~\ref{alg:fpt_plus}, we essentially extend Algorithm~\ref{alg:fpt} in a recursive fashion such that it does not only guess $O_H$ but also guesses $O_L^{(1)}$ afterwards, and once we have the right guess of $O_H\cup O_L^{(1)}$, we can repeat the inner loop of Algorithm~\ref{alg:fpt} and apply similar analysis (this time $O_L^{(0)}$ plays the similar role to $O_L$ in the previous analysis), and analogous to the previous analysis, we can partition $O_L^{(0)}$ into the part selected by the algorithm and the part that is not selected and then repeat the process again by guessing the selected part of $O_L^{(0)}$. Therefore, if $O_L^{(1)}$ (or its analogy in the later repetition) has a lot of value, we can completely extract its value by repeating the process, and if its value is negligible, we enter the regime where almost all the value of $O_L$ (or its analogy) is from $O_L^{(0)}$ (or its analogy), and this regime allows better approximation ratio.

The runtime of Algorithm~\ref{alg:fpt_plus} is still $n2^{\tilde{O}(k)}$. Initially, $I=S_{3k}\cup H$ has size $\tilde{O}(k^2)$, and each recursive call adds $\le Tk$ elements to $I$. Since the depth of recursion is $T$, and we will only need $T$ to be sufficiently large constant, it follows that $|I|=\tilde{O}(k^2)$ in all the recursive calls. Notice that each call of Algorithm~\ref{alg:fpt_recursive} on input $(f, E, k, I, t, T)$ incurs $\le Tk \binom{|I|}{k}=2^{\tilde{O}(k)}$ many recursive calls. Because the depth of recursion is $T$, there are $2^{\tilde{O}(k)}$ calls of Algorithm~\ref{alg:fpt_recursive} in total. In each call, we run greedy algorithm and exhaustive search (Line~\ref{algline:brute_force_recursive}) at most $2^{\tilde{O}(k)}$ times. Greedy algorithm takes time $O(n k)$, and exhaustive search takes time $2^{\tilde{O}(k)}$. It follows that the total runtime is $n2^{\tilde{O}(k)}$. Now we prove the improved approximation factor.

\begin{algorithm}
\caption{{\sc Recursive}$(f, E, k, I, t, T)$ \label{alg:fpt_recursive}}
\begin{algorithmic}[1]
\State Initialize empty sets $I'$ and $X^*$.
\For{each size-$(\le k)$ subset $O_t^{\textrm{guess}}\subseteq I$}
    \For{$i=1,2,\ldots,T$}
        \State Run greedy algorithm on $E$ with $O_t^{\textrm{guess}}$ as the initial solution set to build a size-$k$ solution set $O_t^{\textrm{guess}}\cup X_i$.
        \State Add $X_i$ to $I'$ and remove $X_i$ from $E$.
    \EndFor
    \State Add $I'$ back to $E$.
    \If{$t\le T$}
        \State Run Algorithm~\ref{alg:fpt_recursive} on input $(f, E, k, I\cup I', t+1, T)$, which returns solution set $X'$.
        \State Let $X^*=X'$ if $f(X')>f(X^*)$.
    \Else
        \State Let $X^*=\argmax_{X\subseteq I,\,|X|\le k} f(X)$.\label{algline:brute_force_recursive}
    \EndIf
\EndFor
\State \textbf{return} $X^*$
\end{algorithmic}
\end{algorithm}

\begin{algorithm}
\caption{{\sc FPT$^{+}$}$(f, E, k, \eps, T)$ \label{alg:fpt_plus}}
\begin{algorithmic}[1]
\State Initialize an empty set of elements $I$.
\State Run Algorithm~\ref{alg:symmetric} on input $(f, E, k, \eps)$ and keep $S_{3k}$ and the final version of $H$ in Algorithm~\ref{alg:symmetric}.
\State Run Algorithm~\ref{alg:fpt_recursive} on input $(f, E, k, S_{3k}\cup H, 1, T)$, which returns $X^*$.
\State \textbf{return} $X^*$
\end{algorithmic}
\end{algorithm}

\begin{theorem}\label{thm:fpt_plus}
For sufficiently large constant $T$ and sufficiently small constant $\eps$, Algorithm~\ref{alg:fpt_plus} achieves $0.539$-approximation for non-negative non-monotone submodular maximization with a cardinality constraint.
\end{theorem}
\begin{proof}
Notice that the first call of Algorithm~\ref{alg:fpt_recursive} in Algorithm~\ref{alg:fpt_plus} essentially does the same thing as Algorithm~\ref{alg:fpt} except that Algorithm~\ref{alg:fpt_recursive} makes recursive calls. Thus, we think of the first call of Algorithm~\ref{alg:fpt_recursive} as Algorithm~\ref{alg:fpt} and follow the notations and the setup in the proof of Theorem~\ref{thm:fpt}. We focus on the iteration of the outer loop in the first call of Algorithm~\ref{alg:fpt_recursive} where $O_t^{\textrm{guess}}=O_H$. Same as the proof of Theorem~\ref{thm:fpt}, we divide $O_L$ into $O_L^{(1)}$, the part that is selected during that outer iteration, and $O_L^{(0)}$, the part that is not selected. In the recursive call issued by the first call, we can think of $O_t^{\textrm{guess}}$ as the guess\footnote{One might notice that we do not need to guess $O_H$ again in the recursive call, because we can simply let the first call pass its guess $O_H$ to the recursive call. However, this optimization does not improve the runtime asymptotically.} of $O_H\cup O_L^{(1)}$, and we focus on the outer iteration of the recursive call where $O_t^{\textrm{guess}}=O_H\cup O_L^{(1)}$. Then, we divide $O_L^{(0)}$ into $O_L^{(0,1)}$, the part selected during the outer iteration of the recursive call, and $O_L^{(0,0)}$, the part that is not selected. In the next recursive call issued by this recursive call, $O_t^{\textrm{guess}}$ is the guess of $O_H\cup O_L^{(1)}\cup O_L^{(0,1)}$. We focus on the outer iteration where the guess is correct, and we can continue the analogous analysis for the later recursive calls.

Consider the path of recursions that make correct guesses. There are $T$ such recursive calls. Notice that in each recursive call with correct guess, we extract a new part of $O_L$, which has certain marginal value to the set guessed in that call. Observe that either there is one of them in which the extracted new part of $O_L$ has marginal value less than $f(O_L|O_H)/T$ or we extract $f(O_L|O_H)$ completely and therefore achieve optimality in the last recursion. Obviously, in the second case, we are done, and hence, we focus on the first case and consider the first recursion $t^*$ in which the extracted new part of $O_L$ has marginal value less than $f(O_L|O_H)/T$.
\subsubsection*{Factor-revealing program}
Let $O_L^{(1,\textrm{prev})}$ be the part of $O_L$ extracted in the previous recursions, let $O_L^{(1,\textrm{now})}$ be the newly extracted part of $O_L$ in the recursion $t^*$, and let $O_L^{(0,\textrm{now})}$ be the part of $O_L$ that has not been extracted. Analogous to the proof of Theorem~\ref{thm:fpt}, we let $\alpha_i:=f(X_i|O_H\cup O_L^{(1,\textrm{prev})})$ where $X_i$ (defined in the pseudocode of Algorithm~\ref{alg:fpt_recursive}) is selected in the recursion $t^*$, and moreover, we let $\alpha:=\sum_{i\in[T]}\alpha_i/T$, and furthermore, we let $\beta_0:=f(O_L^{(0,\textrm{now})}|O_H\cup O_L^{(1,\textrm{prev})})$ and $\beta_1:=f(O_L^{(1,\textrm{now})}|O_H\cup O_L^{(1,\textrm{prev})})$. Following the same analysis as the proof of Theorem~\ref{thm:fpt} (by changing $O_H$ to $O_H\cup O_L^{(1,\textrm{prev})}$, $O_L^{(0)}$ to $O_L^{(0,\textrm{now})}$, and $O_L^{(1)}$ to $O_L^{(1,\textrm{now})}$), we get the system of constraints for the factor revealing program, which is either Eq.~\eqref{eq:symmetric_stream_1}~\eqref{eq:symmetric_stream_3}~\eqref{eq:fpt_new_constraint_1}~\eqref{eq:fpt_new_constraint_2}~\eqref{eq:fpt_new_constraint_3} or Eq.~\eqref{eq:symmetric_stream_1}~\eqref{eq:symmetric_stream_4}~\eqref{eq:fpt_new_constraint_1}~\eqref{eq:fpt_new_constraint_2}~\eqref{eq:fpt_new_constraint_3}, and we get the same objective of the factor revealing program $r_{\textrm{FPT}}$.

Finally, we add one more constraint (this is the point of having a new algorithm). That is
\begin{equation}\label{eq:fpt_plus_constraint}
    \beta_1\approx 0,
\end{equation}
which holds because we are in the case where $f(O_L^{(1,\textrm{now})}|O_H\cup O_L^{(1,\textrm{prev})})\le f(O_L|O_H)/T\le 1/T$, and $T$ can be an arbitrarily large constant.

We can find the lower bound of approximation ratio by numerically solving the following two convex programs,
$$\min_{a,b,c,d,\alpha,\beta_0,\beta_1} r_{\textrm{FPT}}(a,b,c,d,\alpha,\beta_0,\beta_1) \textrm{ s.t. Eq.~\eqref{eq:symmetric_stream_1}~\eqref{eq:symmetric_stream_3}~\eqref{eq:fpt_new_constraint_1}~\eqref{eq:fpt_new_constraint_2}~\eqref{eq:fpt_new_constraint_3}~\eqref{eq:fpt_plus_constraint}},$$
$$\min_{a,b,c,d,\alpha,\beta_0,\beta_1} r_{\textrm{FPT}}(a,b,c,d,\alpha,\beta_0,\beta_1) \textrm{ s.t. Eq.~\eqref{eq:symmetric_stream_1}~\eqref{eq:symmetric_stream_4}~\eqref{eq:fpt_new_constraint_1}~\eqref{eq:fpt_new_constraint_2}~\eqref{eq:fpt_new_constraint_3}~\eqref{eq:fpt_plus_constraint}},$$
and the minimum of the two results turns out to be larger than $0.539$.
\end{proof}

\section{$(1/2+c)$-approximation for random-order streaming symmetric submodular maximization}
In this section, we show that Algorithm~\ref{alg:symmetric} \emph{beats} $1/2$-approximation for symmetric non-monotone submodular functions using $\widetilde{O}(k^2)$ memory. Together with our lower bound result (Theorem~\ref{thm:lower_bound}), this separates the symmetric non-monotone submodular functions from general non-monotone submodular functions in the random-order streaming model. To our best knowledge, this is first such separation.

The following lemma is the key feature of symmetric submodular functions which we will take advantage of, and it basically says for symmetric submodular function, a set can not hurt another set by more than its own value.
\begin{lemma}\label{lem:symmetric_cannot_hurt_much}
For any non-negative symmetric submodular function $f:V\to\RR_{\ge0}$, for any disjoint $X,Y\subseteq V$, $f(X|Y)\ge -f(X)$.
\end{lemma}
\begin{proof}
By submodularity, $f(X|Y)\ge f(X|V\setminus X)=f(V)-f(V\setminus X)$, and by symmetry and non-negativity, $f(V)-f(V\setminus X)=f(\emptyset)-f(X)\ge-f(X)$.
\end{proof}
Before going to the technical proof of the better-than-1/2 approximation, we give the interpretable analysis for why Algorithm~\ref{alg:symmetric} can beat $1/2$ for symmetric submodular functions. The interpretable analysis is still a little lengthy and technical. At very high level, the idea is if none of $S_{|O_L|}$ and $O_H\cup S_{|O_L|}$ and $O_H\cup (S_{2|O_L|}\setminus S_{|O_L|})$ beats $1/2$, then we can show that (i) $f(O_H)=f(O_L)=\frac{f(O_H\cup O_L)}{2}$, (ii) $f(O_H|S_{2|O_L|})=-f(O_H)$, and (iii) $f(S_{3|O_L|}\setminus S_{2|O_L|}|S_{2|O_L|})\ge (1-1/e)f(O_L)$. The punchline is that using (ii), we can further show that (iv) $f(S_{3|O_L|}\setminus S_{2|O_L|}|O_H)\ge f(S_{3|O_L|}\setminus S_{2|O_L|}|S_{2|O_L|})$ (basically, the argument is if $O_H$ hurts $S_{2|O_L|}$ a lot, then by Lemma~\ref{lem:symmetric_cannot_hurt_much}, which is due to symmetry, one can argue $O_H$ can not hurt $S_{3|O_L|}$ more than how much it hurts $S_{2|O_L|}$). Therefore, by (i) and (iii) and (iv), we know that $O_H\cup(S_{3|O_L|}\setminus S_{2|O_L|})$ beats $1/2$ approximation. Now we explain this idea in more details.

\paragraph{Informal interpretable analysis} The starting point is the analysis for Theorem~\ref{thm:non-monotone_half_apx}. The reader can first review the intuition given in the beginning of the subsection of Theorem~\ref{thm:non-monotone_half_apx}. There we argued that the algorithm achieves half of $f(O_H\cup S_{|O_L|})+f(S_{|O_L|})\ge f(O_H\cup O_L)$, and hence for the instance to be hard (in the sense that the algorithm only gets $1/2$ approximation), it requires $f(O_H\cup S_{|O_L|})=f(S_{|O_L|})=\frac{f(O_H\cup O_L)}{2}$. This implies $f(O_H|S_{|O_L|})=0$, and by submodularity $f(O_H|O_L\cup S_{|O_L|})\le0$, and hence $f(O_H\cup O_L\cup S_{|O_L|})\le f(O_L\cup S_{|O_L|})$. Recall that we argued $S_{|O_L|}$ can not hurt $O_H\cup O_L$ significantly, and thus, $f(O_H\cup O_L)\le f(O_H\cup O_L\cup S_{|O_L|})\le f(O_L\cup S_{|O_L|})$, but since $f(S_{|O_L|})=\frac{f(O_H\cup O_L)}{2}$, we have that $f(O_L|S_{|O_L|})=f(S_{|O_L|})$, which implies\footnote{Intuitively, by the if condition at line~\ref{algline:memory-threshold} of Algorithm~\ref{alg:symmetric}, we can show that the marginal contribution of each iterate of $S_{|O_L|}$ is at least $\frac{f(O_L|S_{|O_L|})}{|O_L|}$, but because $f(O_L|S_{|O_L|})=f(S_{|O_L|})$, each iterate actually makes the same marginal contribution. Notice that the first iterate should make contribution more than any element in $O_L$ by the if condition.} $f(S_{|O_L|})\ge f(O_L)$. Moreover, since $f(S_{|O_L|})=\frac{f(O_H\cup O_L)}{2}$, we have $f(O_L)\le \frac{f(O_H\cup O_L)}{2}$ and hence $f(O_H)\ge \frac{f(O_H\cup O_L)}{2}$, but we also have $f(O_H)\le \frac{f(O_H\cup O_L)}{2}$ because otherwise $O_H\cup S_{|O_L|}$ should have beaten $1/2$-approximation as $S_{|O_L|}$ does not hurt $O_H$ significantly, and therefore 
it holds that $f(O_H)=f(O_L)=\frac{f(O_H\cup O_L)}{2}$.

Now consider the set $S_{2|O_L|}\setminus S_{|O_L|}$. If $f(S_{2|O_L|}\setminus S_{|O_L|}|O_H)=0$ (otherwise $S_{2|O_L|}\setminus S_{|O_L|}\cup O_H$ beats $1/2$), then by submodularity and the fact that $S_{2|O_L|}\setminus S_{|O_L|}$ does not hurt anything significantly (which is yet another application of Lemma~\ref{lem:first_eps_does_not_hurt}), we have $f(S_{2|O_L|}\setminus S_{|O_L|}|O_H\cup S_{|O_L|})=0$ and hence $f(O_H\cup S_{2|O_L|})=f(O_H\cup S_{|O_L|})=\frac{f(O_H\cup O_L)}{2}$. Notice that $f(O_L|O_H\cup S_{2|O_L|})=f(O_L\cup O_H\cup S_{2|O_L|})-f(O_H\cup S_{2|O_L|})\ge f(O_L\cup O_H)-f(O_H\cup S_{2|O_L|})=\frac{f(O_H\cup O_L)}{2}$, where the inequality is again due to the fact that $S_{2|O_L|}$ does not hurt. Therefore, similar to how we argued $f(S_{|O_L|})\ge f(O_L)$, we can show that $f(S_{2|O_L|}\setminus S_{|O_L|}|S_{|O_L|})\ge f(O_L)$. Since $f(S_{2|O_L|})\ge f(S_{2|O_L|}\setminus S_{|O_L|}|S_{|O_L|})+f(S_{|O_L|})\ge 2f(O_L)=2f(O_H)$ and $f(O_H\cup S_{2|O_L|})=\frac{f(O_H\cup O_L)}{2}=f(O_H)$, we have that $f(O_H|S_{2|O_L|})\le -f(O_H)$, and together with Lemma~\ref{lem:symmetric_cannot_hurt_much}, we have that $f(O_H|S_{2|O_L|})=-f(O_H)$.

Here comes the final punchline---If $S_{3|O_L|}/S_{2|O_L|}$ has significant marginal contribution to $S_{2|O_L|}$, then $S_{3|O_L|}/S_{2|O_L|}$ must have at least the same marginal contribution to $O_H$ (and therefore, $O_H\cup S_{3|O_L|}/S_{2|O_L|}$ will beat $1/2$). Specifically, this follows from
\begin{align*}
    &f(S_{3|O_L|}/S_{2|O_L|}|O_H)\\&\ge f(S_{3|O_L|}/S_{2|O_L|}|O_H\cup S_{2|O_L|})&&\text{(By submodularity)}\\
    &=f(O_H|S_{3|O_L|})-f(O_H|S_{2|O_L|})+f(S_{3|O_L|}/S_{2|O_L|}|S_{2|O_L|})\\
    &\ge-f(O_H)-f(O_H|S_{2|O_L|})+f(S_{3|O_L|}/S_{2|O_L|}|S_{2|O_L|})&&\text{(By Lemma~\ref{lem:symmetric_cannot_hurt_much})}\\
    &=f(S_{3|O_L|}/S_{2|O_L|}|S_{2|O_L|})&&\text{(By $f(O_H|S_{2|O_L|})=-f(O_H)$).}
\end{align*}
(One can check the first equality is true by expanding both sides of the equality.) It remains to show $f(S_{3|O_L|}/S_{2|O_L|}|S_{2|O_L|})$ is indeed significantly large. This is essentially due to $f(O_L|O_H\cup S_{2|O_L|})\ge\frac{f(O_H\cup O_L)}{2}$, which we argued earlier, and the if condition at line~\ref{algline:memory-threshold}. In particular, we can show $f(S_{3|O_L|}/S_{2|O_L|}|S_{2|O_L|})$ is at least $1-1/e$ fraction of $f(O_L|S_{2|O_L|})$ by the standard analysis of the classic greedy algorithm for monotone submodular maximization.

Now we formally prove the better-than-$1/2$ approximation of Algorithm~\ref{alg:symmetric} using the factor-revealing programs.

\begin{theorem}\label{thm:symmetric_beating_half}
For sufficiently small $\eps$ and large $k$, Algorithm~\ref{alg:symmetric} achieves strictly better-than-$1/2$ approximation for non-negative non-monotone symmetric submodular functions in the random-order streaming model.
\end{theorem}
\begin{proof}
The proof is based on factor-revealing convex programs, and the setup of the proof is already given in Subsection~\ref{subsection:common_setup}. Recall that in Lemma~\ref{lem:factor_revealing_common_constraints}, we have established a system of constraints that is either Eq.~\eqref{eq:symmetric_stream_1} and \eqref{eq:symmetric_stream_3}, or Eq.~\eqref{eq:symmetric_stream_1} and \eqref{eq:symmetric_stream_4}. Now we establish the objective for the programs.

\subsubsection*{Objective for factor-revealing program}
Our goal is to lower bound $f(S_L''|O_H)$. To this end, we want to lower bound $f(S_L''|S_L\cup S_L')$ and upper bound $f(O_H|S_L\cup S_L')$. First, we lower bound $f(O_L|S_L\cup S_L')$ as follows,
\begin{align}
    &f(O_L|S_L\cup S_L')\ge f(O_L|S_L\cup S_L'\cup O_H) &&\text{(By submodularity)} \nonumber\\
    &\gtrsim f(O_L\cup O_H)-f(S_L\cup S_L'\cup O_H) &&\text{(By event $A_1$ and submodularity)} \nonumber\\
    &=f(O_L\cup O_H)-f(S_L\cup O_H)-f(S_L'|S_L\cup O_H) \nonumber\\
    &\ge f(O_L\cup O_H)-f(S_L\cup O_H)-f(S_L'|O_H) &&\text{(By submodularity)} \nonumber\\
    &\approx 1-a-b-d. &&\text{(By event $A_3$)}\label{eq:symmetric_stream_5}
\end{align}
Notice that analogous to Eq.~\eqref{eq:symmetric_stream_0.5}, we have that
\begin{align}
    f(S_L'|S_L)&=\sum_{i=1}^{\ell} f(e_{i+\ell}|S_{\ell+i-1}) && \text{(By telescoping sum)}\nonumber\\
    &\ge\sum_{i=1}^{\ell} \frac{\sum_{o\in O_L}f(o|S_{\ell+i-1})}{|O_L|} && \text{(By Observation~\ref{obs:O_L_is_inferior})}\nonumber\\
    &=\sum_{i=1}^{\ell} \frac{\sum_{o\in O_L}f(o|S_L\cup S_L')}{|O_L|} && \text{(By submodularity)} \nonumber\\
    &\le \sum_{i=1}^{\ell} \frac{f(O_L|S_L\cup S_L')}{|O_L|} && \text{(By submodularity)} \nonumber\\
    &\ge f(O_L|S_L\cup S_L'). && \text{(By $\ell\ge|O_L|$)}
    \label{eq:symmetric_stream_5.5}
\end{align}
It follows by Eq.~\eqref{eq:symmetric_stream_5} and Eq.~\eqref{eq:symmetric_stream_5.5} that $f(S_L'|S_L)\gtrsim1-a-b-d$. We apply this inequality to derive,
\begin{align}
    &f(O_H|S_L\cup S_L')=f(O_H\cup S_L\cup S_L')-f(S_L\cup S_L') \nonumber\\
    &\le f(S_L'|O_H)+f(S_L\cup O_H)-f(S_L)-f(S_L'|S_L) &&\text{(By submodularity)} \nonumber\\
    &= f(S_L'|O_H)+f(O_H|S_L)-f(S_L'|S_L) \nonumber\\
    &= a+d-f(S_L'|S_L) &&\text{(By definition of $a,d$)} \nonumber\\
    &\lesssim 2a+b+2d-1. &&\text{(By $f(S_L'|S_L)\gtrsim1-a-b-d$)} \label{eq:symmetric_stream_6}
\end{align}
On the other hand, for any $2\ell+1\le i\le 3\ell$, it follows by event $A_1$ that (specifically, this uses the second bullet point of Lemma~\ref{lem:first_eps_does_not_hurt} by noticing that $S_L\cup S_L'=S_{2\ell}$ and $S_{i-1}\subseteq S_{3\ell}$)
$$f(S_{i-1}\setminus(S_L\cup S_L')|O\cup S_L\cup S_L')\gtrsim 0,$$
and therefore, we have that
\begin{align*}
    &f(e_i|S_{i-1})\ge \frac{\sum_{o\in O_L} f(o|S_{i-1})}{|O_L|} &&\text{(By Observation~\ref{obs:O_L_is_inferior})} \nonumber\\
    &\ge \frac{f(O_L|S_{i-1})}{|O_L|} &&\text{(By submodularity)} \nonumber\\
    &=\frac{f(S_{i-1}\setminus(S_L\cup S_L')|O_L\cup S_L\cup S_L')+f(O_L\cup S_L\cup S_L')-f(S_{i-1})}{|O_L|} \nonumber\\
    &\ge\frac{f(S_{i-1}\setminus(S_L\cup S_L')|O\cup S_L\cup S_L')+f(O_L\cup S_L\cup S_L')-f(S_{i-1})}{|O_L|} \nonumber&&\text{(By submodularity)}\\
    &\gtrsim \frac{f(O_L\cup S_L\cup S_L')-f(S_{i-1})}{|O_L|} \nonumber\\
    &=\frac{f(O_L| S_L\cup S_L')-f(S_{i-1}\setminus(S_L\cup S_L')|S_L\cup S_L')}{|O_L|}.
\end{align*}
Using the inequality above (which can be thought of as the improvement guarantee of one iteration of greedy algorithm when the objective function is $f(X|S_L\cup S_L')$) and standard analysis for the greedy algorithm, we can lower bound $f(S_L''|S_L\cup S_L')$ by $(1-1/e)f(O_L|S_L\cup S_L')$. 

We are ready to lower bound $f(S_L''|O_H)$,
\begin{align*}
    &f(S_L''|O_H)\ge f(S_L''|O_H\cup S_L\cup S_L') &&\text{(By submodularity)} \nonumber\\
    &=f(O_H|S_L''\cup S_L'\cup S_L)-f(O_H|S_L'\cup S_L)+f(S_L''|S_L\cup S_L') \nonumber\\
    &\ge -f(O_H)-f(O_H|S_L'\cup S_L)+f(S_L''|S_L\cup S_L') &&\text{(By Lemma~\ref{lem:symmetric_cannot_hurt_much})} \nonumber\\
    &\gtrsim -f(O_H)-f(O_H|S_L'\cup S_L)+(1-1/e)f(O_L|S_L\cup S_L').
\end{align*}
It follows that
\begin{align}
    &f(S_L''\cup O_H)\ge -f(O_H|S_L'\cup S_L)+(1-1/e)f(O_L|S_L\cup S_L') \nonumber\\
    &\gtrsim (2-1/e)(1-a-b-d)-d-a. &&\text{(By Eq.~\eqref{eq:symmetric_stream_5} and~\eqref{eq:symmetric_stream_6})}
\end{align}
Finally, observe that all of $S_L,S_L\cup O_H, S_L'\cup O_H, S_L''\cup O_H$ have size at most $(1+\eps^{1/3})k$, and by Lemma~\ref{lem:subsample}, they all have $\frac{1}{1+\eps^{1/3}}$-approximate size-$k$ subsets, which are candidate solutions in exhaustive search, and therefore, the value algorithm achieves is at least $\frac{r(a,b,c,d)}{1+\eps^{1/3}}$ (up to an error of $\textrm{poly}(\eps+1/k)$), where
\begin{equation*}
    r(a,b,c,d):=\max\{b,\,a+b,\,d+c,\,(2-1/e)(1-a-b-d)-d-a\}.
\end{equation*}
Hence, $r(a,b,c,d)$ is essentially a lower bound for the approximation ratio, given $a,b,c,d$, and therefore, we can numerically solve the following two convex programs,
$$\min_{a,b,c,d} r(a,b,c,d) \textrm{ s.t. Eq.~\eqref{eq:symmetric_stream_1} and~\eqref{eq:symmetric_stream_3}},$$
$$\min_{a,b,c,d} r(a,b,c,d) \textrm{ s.t. Eq.~\eqref{eq:symmetric_stream_1} and~\eqref{eq:symmetric_stream_4}},$$
and the minimum of two results gives a lower bound for the approximation ratio. The minimum turns out to be larger than $0.5029$.
\end{proof}

On a side note, the constant we get here is by no means tight. (Indeed, we have an improvement, which may also improve the constant for our FPT algorithm, but it requires numerically solving non-convex programs.) What is interesting is the separation between symmetric and general submodular functions in the random-order streaming model. Also, it is tempting to conjecture that Algorithm~\ref{alg:symmetric} achieves optimal $1-1/e$ approximation for monotone submodular functions, given its success in the non-monotone regime. Nonetheless, we have a hard instance that refutes this conjecture. The details would be made available to the interested reader upon request.

\section{Tight 1/2 hardness for random-order streaming non-monotone submodular maximization}
In this section, we present the lower bound result for non-monotone submodular maximization in the random-order streaming model (described in Section~\ref{section:prelim}). The approximation factor in the lower bound result is tight because of the upper bound in Theorem~\ref{thm:non-monotone_half_apx} for example.
\begin{theorem}\label{thm:lower_bound}
Assuming $n=2^{o(k)}$, any $(1/2+\eps)$-approximation algorithm for non-monotone submodular maximization in the random-order streaming model must use $\Omega(n/k^2)$ memory.
\end{theorem}
\begin{proof}
\subsubsection*{Construction of the hard instance}
The function $f$ we construct here is essentially a cut function on an unweighted bipartite directed hypergraph\footnote{A directed hyperedge in a directed hypergraph is represented by some $(U,v)$, where $U$ is a subset of vertices, and $v\notin U$ is a vertex. For any subset of vertices $S$, a hyperedge $(U,v)$ is cut by $S$ iff $|U\cap S|>0$ and $v\notin S$. It is well-known that such cut function is submodular.} plus a modular function.
The ground set $V:=A_1\cup A_2$ for $f$ is the set of $n$ vertices of the graph, where $A_1$ and $A_2$ denote the two parts respectively. Specifically, $A_2:=\{u_1,\dots,u_{\eps k}\}$, and $A_1$ is partitioned into $\ell:=(n-\eps k)/b$ buckets of vertices $B_1,\dots,B_{\ell}$, each of size $b:=k-\eps k$. Now we describe a random generating procedure that generates the hyperedges in the graph:
\begin{enumerate}
    \item First, for each $i\in[\ell]$, we sample a random subset of vertices $N_i \subset A_2$ of size $|N_i|=\eps^2k$, and for each $u_j\in N_i$, we create a directed hyperedge from $B_i$ to $u_j$.
    \item Then, we slightly modify the graph generated in step 1 as follows: We sample a uniformly random $g\in[\ell]$. For each $u_j\in N_{g}$, we remove the hyperedge from $B_{g}$ to $u_j$, and instead, for each $v\in B_{g}$, we create a directed hyperedge from $\{v\}$ to $u_j$. (That is, for each $u_j\in N_{g}$, we replace the hyperedge from $B_{g}$ to $u_j$ with individual edges from each $v\in B_{g}$ to $u_j$.)
\end{enumerate}
The final submodular function $f:V\to\RR_{\ge 0}$ is the sum of the cut function on the above generated hypergraph plus the modular function $c(S):=(\eps^2 k^2)\cdot \frac{|S\cap A_2|}{|A_2|}$. See Figure~\ref{fig:hard_instance} for an illustration.
\begin{figure}[h]
\centering
    \includegraphics[scale=0.5]{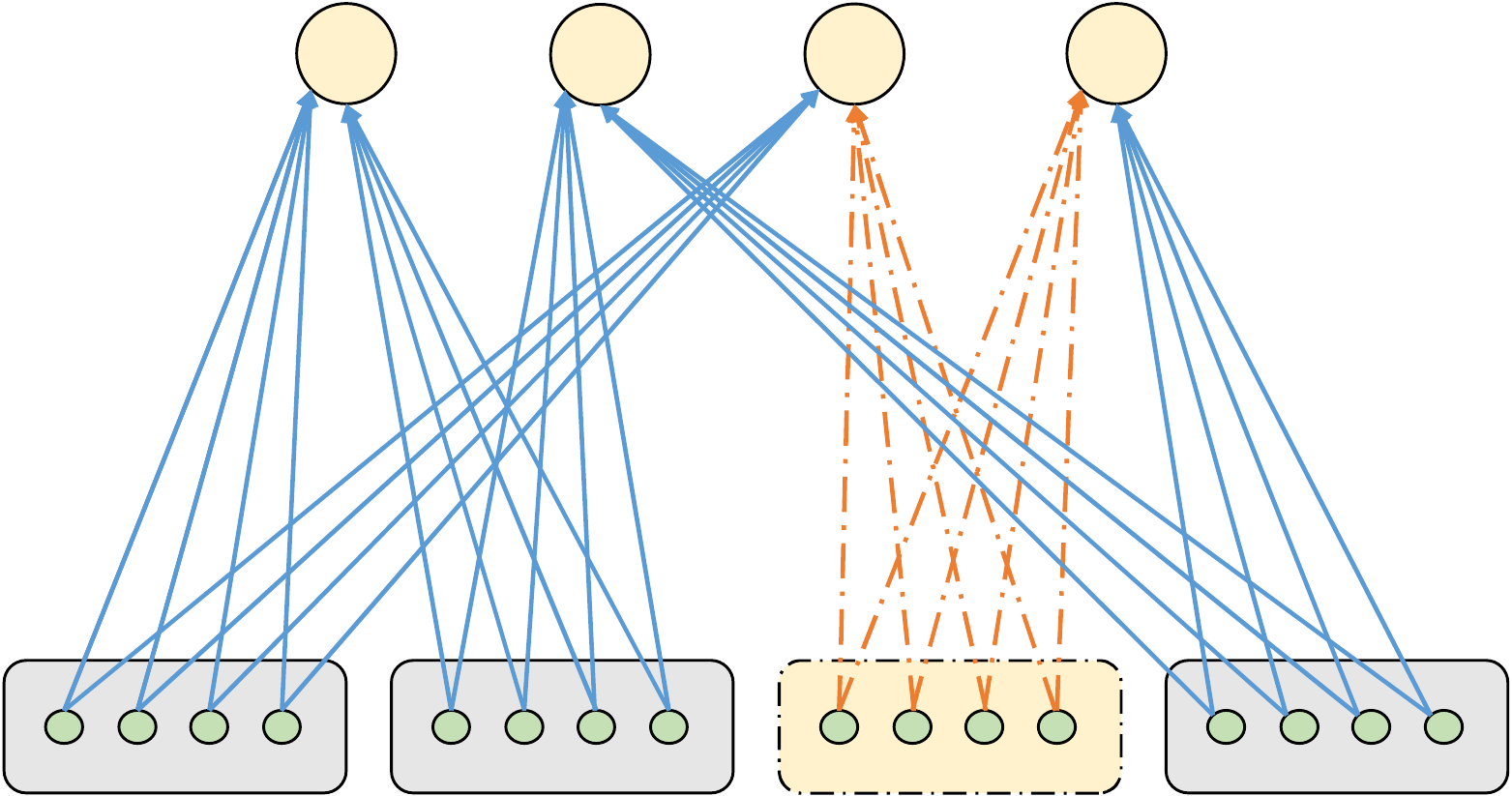}
    \caption{An illustration of our hard instance: At the top we have all $\eps k$ vertices of $A_2$, each of which has value $\eps k$. At the bottom we have all $(n-\eps k)$ vertices of $A_1$ that are separated into $\ell=(n-\eps k)/b$ buckets, each of which has $b=k-\eps k$ vertices. We choose a bucket $B_g$ (with yellow filling and dashed outline) uniformly at random. There are $\eps k$ individual edges (orange and dashed) from each vertex in bucket $B_g$ to $B_g$'s neighborhood $N_g$, and there are $\eps k$ hyperedges (blue and solid) from every other bucket $B_i$ (with gray filling and solid outline) to its neighborhood $N_i$.}
    \label{fig:hard_instance}
\end{figure}
\subsubsection*{Setting up the hardness proof}
Before proving hardness of our hard instances, we first grant the algorithm some extra power (we will prove hardness for these stronger algorithms):
\begin{enumerate}
    \item We let the algorithm keep all the elements in $A_2$ in its memory for free at the beginning of the stream, without memory cost.
    \item We reveal to the algorithm the graph topology generated in step 1 of the generating procedure at the beginning of the stream, without memory cost. To be more precise, at the beginning of the stream, we let the algorithm know the sets $B_i$ and $N_i$ for all $i\in[\ell]$, but at this point, the algorithm does not have the actual elements of any $B_i$ in its memory (and hence cannot query\footnote{\label{foot:random_order_model}Recall that in the random-order streaming model, the algorithm is only allowed to use any subset of elements in its memory as input to the value oracle of the objective submodular function.} the value of any set that intersects $B_i$), and moreover, the algorithm does not know which $g$ we chose to make the modification in step 2 of the generating procedure.
    \item The algorithm can store infinite amount of information (but not elements) during the stream without memory cost. To be more precise, at any time during the stream, besides the elements in $A_2$, the algorithm with memory $m$ is only allowed to store a set of $m$ elements $S$, and it is only allowed to query\footref{foot:random_order_model} the value of any subset of $S\cup A_2$, but it is allowed to store the result of the query forever without memory cost.
\end{enumerate}
Given these extra power, it is not hard to see that at the beginning of the stream, the algorithm already knows the value of any set that has intersection of size $\le 1$ with each $B_i$ (and what is unknown to the algorithm is for each $i\in[\ell]$, whether $i=g$). Therefore, in order to get any new information during the stream, the algorithm has to query the value of a set that has intersection of size $\ge 2$ with some $B_i$, and to be able to make such query, the algorithm has to store $\ge 2$ elements of $B_i$ together in its memory. On the other hand, if the algorithm has $\ge 2$ elements $j_1,j_2$ of $B_i$ together in its memory, it can immediately tell whether $i=g$ (specifically, by querying the value of set $\{j_1,j_2\}$, the algorithm can tell whether there are hyperedges or individual edges from $B_i$ to $N_i$).
This motivates us to introduce a definition that is useful for formally phrasing these observations:
\begin{definition}\label{def:undetected}
At any time during the stream, for each $i\in[\ell]$, we say $B_i$ has been \underline{detected} if $\ge2$ elements of $B_i$ were stored in the algorithm's memory together at some point in the past, and we say $B_i$ is \underline{undetected} if otherwise.
\end{definition}
Our discussion before Definition~\ref{def:undetected} can now be generalized and phrased more formally:
\begin{observation}\label{obs:undetected}
At any time during the stream, let $U\subseteq [\ell]$ be the set such that $i\in U$ iff $B_i$ is undetected, and assume that $g\in U$, then the algorithm knows $j\neq g$ (and knows that the edges from $B_j$ to $N_j$ are hyperedges) for all $j\notin U$, but it has no information about which $i\in U$ is $g$ (unless $|U|=1$), i.e., each $i\in U$ is equally likely to be $g$, conditioned on all the information the algorithm has up to this time and our assumption that $g\in U$.
\end{observation}
\begin{proof}[Proof of Observation~\ref{obs:undetected}]
Since $j\notin U$ and we assumed that $g\in U$, it holds that $j\neq g$. Moreover, as we discussed before Definition~\ref{def:undetected}, once $B_i$ is detected, the algorithm can tell whether $i=g$ by making two queries. For each $j\notin U$, by definition of $U$, $B_j$ has been detected, and thus, the algorithm knows $j\neq g$.

Now, we show that the algorithm has no information about which $i\in U$ is $g$. To see this, we observe that (i) by definition of $U$, the algorithm could never query\footnote{Note that we use the word ``query'' instead of ``know'' because the algorithm knew (from the very beginning) the value of some uninteresting sets that contain $\ge 2$ elements $j_1,j_2$ of $B_i$ for $i\in U$, e.g., the value of $\{j_1,j_2\}\cup A_2$.} the value of any set that has intersection of size $\ge 2$ with $B_i$ for any $i\in U$, and (ii) by our construction, the value of any set, that has intersection of size $\le 1$ with $B_i$ for all $i\in U$, is independent of which $i\in U$ was chosen as $g$ at step 2 of the generating procedure. Therefore, the result of any query, that the algorithm could possibly make, is independent of which $i\in U$ was chosen as $g$. Moreover, since we chose $g\in[\ell]$ uniformly at random in the generating procedure, it follows that each $i\in U$ is equally likely to be $g$, conditioned on the results of all the queries that the algorithm could possibly make and the information we revealed to the algorithm at the beginning of the stream.
\end{proof}

Moreover, by Definition~\ref{def:undetected}, in order to be able to detect $B_i$, the algorithm has to store an element $e'\in B_i$ in the memory when a different element $e\in B_i$ arrives in the stream, which motivates the following definition:
\begin{definition}\label{def:collision}
When an element $e\in V$ in the stream arrives, suppose the algorithm has stored a set of elements $S\subseteq V$ in its memory, then for any $i\in[\ell]$, if $e\in B_i$ and $|S\cap B_i|>0$, we say there is a \underline{collision} in $B_i$. To be more precise, the collision occurs as soon as $e$ arrives (before the algorithm stores $e$ or makes any new query).
\end{definition}
Notice that in Definition~\ref{def:collision}, we made clear that the collision occurs before the algorithm stores the new element or makes any new query, and thus, at the very moment when the collision occurs, the algorithm does not gain any new information (although it can get new information after this moment by storing the new element and making new queries). Therefore, the following observation follows easily from Observation~\ref{obs:undetected}:
\begin{observation}\label{obs:collision}
For any $t\in[n]$, let $U\subseteq [\ell]$ be the set such that $i\in U$ iff $B_i$ is undetected before the $t$-th element of the stream $e\in V$ arrives, and assume that $g\in U$. Suppose that a collision in $B_j$ occurs (for some $j\in[\ell]$) when the $t$-th element of the stream $e$ arrives, then at the moment when this collision occurs, the probability that $j=g$, conditioned on all the information algorithm has and our assumption that $g\in U$, is at most $1/|U|$.
\end{observation}
\begin{proof}[Proof of Observation~\ref{obs:collision}]
If $j\notin U$, clearly the probability that $j=g$ conditioned on $g\in U$ is zero. For the case of $j\in U$, by Observation~\ref{obs:undetected}, each $i\in U$ is equally likely to be $g$, conditioned on all the information algorithm has and our assumption that $g\in U$. Moreover, as we elaborated before Observation~\ref{obs:collision}, at the moment that the collision occurs, the algorithm does not gain any new information, and thus, if $j\in U$, the probability that $j=g$, conditioned on $g\in U$ and all the information the algorithm has, is exactly $1/|U|$.
\end{proof}

\subsubsection*{Outline of the hardness proof}
Before going to the technical proof, we outline our proof strategy:
\begin{enumerate}
    \item First, we use the random-order property of the stream and union bound to show that given any constant $\eta\in(0,1)$, for any algorithm with $o({n}/{k^2})$ memory, w.h.p.~there are at most $o(n/k)$ collisions occurring in total during the stream before the $(1-\eta)n$-th element of the stream arrives.
    \item Then, for any $q<\ell$, we prove by induction (and union bound) and Observation~\ref{obs:collision} that the probability that one of the first $q$ collisions encountered by the algorithm is in $B_g$ is at most $q/(\ell-q)$, which implies that w.h.p.~none of the first $o({n}/{k})$ collisions is in $B_g$.
    
    Note that this together with the first point imply that for any algorithm, w.h.p.~it has never stored $\ge 2$ elements of $B_g$ together in its memory before the $(1-\eta)n$-th element of the stream arrives. However, since by standard concentration inequality, $(\ge1-2\eta)$-fraction of the elements of $B_g$ appear in the first $(1-\eta)$-fraction of the stream w.h.p., it follows that the algorithm missed\footnote{Here we use the fact that the standard streaming model for submodular maximization (described in Section~\ref{section:prelim}) requires the algorithm to output a subset of elements in its memory. In appendix (Section~\ref{section:arbitrary_output}), we sketch how to slightly modify our hard instance to make the hardness result hold against stronger algorithms that are allowed to output any size-$(\le k)$ subset of $V$, which are non-standard but might be of independent interest.} $(1-2\eta)$-fraction of $B_g$ w.h.p.
    \item Finally, we show that in our hard instance, there is always a set that contains $B_g$ and has value $2(1-\eps)\eps^2 k^2$ (completeness), but for any constant $\delta>0$, w.h.p.~every size-$(\le k)$ set that only contains $(\le2\eta)$-fraction of the elements of $B_g$ has value $\le (1+\delta+2\eta)\eps^2k^2$ (soundness), which completes the hardness proof, because $\eta,\delta,\eps$ are constants that can be arbitrarily small.
\end{enumerate}
In the following, we implement the above proof outline.
\subsubsection*{$o({n}/{k^2})$-memory algorithm sees $o(n/k)$ collisions in the first $(1-\eta)$-fraction of the stream}
At any time during the first $(1-\eta)$-fraction of the stream, an algorithm with memory $m=o(\frac{n}{k^2})$ can only store a set of $m$ elements $S$. Let $I:=\{i\mid i\in[\ell]\textrm{ and } |S\cap B_i|>0\}$, i.e., $i\in I$ iff the algorithm stores at least one element of $B_i$, and clearly, we have $|I|\le|S|=m$. By Definition~\ref{def:collision}, a collision in $B_i$ can occur only if there is an element of $B_i$ in the algorithm's memory, and thus, a collision can only occur in some $B_i$ such that $i\in I$. Regardless of what $I$ is, the probability that the next element in the stream causes a collision is at most $\frac{m k}{\eta n}$, because to cause a collision, the next element has to be in some $B_i$ such that $i\in I$, and in total there are only $\le m k$ elements in $\bigcup_{i\in I}B_i$, but there are $\eta n$ elements in the rest of stream arriving in uniformly random order. Moreover, since a new collision can occur only when a new element of the stream arrives, and there are $(1-\eta)n$ elements arriving in the first $(1-\eta)$-fraction of the stream, we have that
\begin{align*}
    &\E[\text{number of collisions in the first $(1-\eta)$-fraction of the stream}]\\
    %\le& (1-\eta)n\cdot\textnormal{Pr}[\text{a collision occurs when a new element arrives}]\\
    \le& (1-\eta)n\cdot\frac{m k}{\eta n}=o(n/k) &&\text{(By $m=o(n/k^2)$)}.
\end{align*}
Finally, by Markov inequality, we have that with probability $1-o(1)$, an $o({n}/{k^2})$-memory algorithm only sees $o(n/k)$ collisions in the first $(1-\eta)$-fraction of the stream.
\subsubsection*{$o({n}/{k^2})$-memory algorithm misses $(\ge1-2\eta)$-fraction of $B_g$}
Given any integer $q<\ell$, we prove by induction that for all $i\le q$, the probability that one of the first $i$ collisions encountered by the algorithm is in $B_g$ is at most $i/(\ell-q)$. The base case of $i=0$ is trivial. For the induction step ($i\ge 1$), let $E_1$ denote the event that one of the first $i-1$ collisions encountered by the algorithm is in $B_g$, then the induction hypothesis is $\textrm{Pr}[E_1]\le(i-1)/(\ell-q)$. Conditioned on the event $\bar{E_1}$ (i.e., none of the first $i-1$ collisions encountered by the algorithm is in $B_g$), we know that $B_g$ must be undetected, i.e., $g\in U$ where $U\subseteq [\ell]$ denotes the set such that $j\in U$ iff $B_j$ is undetected before the $i$-th collision. Moreover, since the algorithm can only detect a new $B_j$ when a collision in $B_j$ occurs (and there are $i-1$ collisions before the $i$-th collision), it follows that $|U|\ge \ell-(i-1)$. Thus, by Observation~\ref{obs:collision}, conditioned on event $\bar{E_1}$, the probability of $E_2$, which denotes the event that the $i$-th collision occurs in $B_g$, is at most $1/|U|\le1/(\ell-(i-1))\le1/(\ell-q)$. Hence, we get $\textrm{Pr}[E_1 \textrm{ or } E_2]=\textrm{Pr}[E_1]+\textrm{Pr}[E_2\mid \bar{E_1}]\le (i-1)/(\ell-q)+1/(\ell-q)=i/(\ell-q)$, which completes the induction step, because ``$E_1 \textrm{ or } E_2$'' is exactly the event that one of the first $i$ collisions encountered by the algorithm is in $B_g$.

Therefore, with probability $1-\frac{o(n/k)}{\ell-o(n/k)}=1-o(1)$, none of the first $o({n}/{k})$ collisions is in $B_g$. As we explained in point 2 of the proof outline, this together with point 1 of the proof outline imply that the algorithm missed $(1-2\eta)$-fraction of $B_g$ w.h.p.

\subsubsection*{The inapproximation factor (completeness and soundness)}
\paragraph{Completeness}
Consider the size-$(\le k)$ set $OPT=B_g\cup (A_2\setminus N_g)$. Notice that this set cuts all the $\eps^2 k\cdot b$ individual edges from $B_g$ to $N_g$ and hence gets value $\eps^2 k\cdot b=(1-\eps)\eps^2 k^2$ from the cut function. Moreover, this set has value $c(OPT)=\eps^2k^2\cdot\frac{|OPT\cap A_2|}{|A_2|}=\eps^2k^2\cdot \frac{\eps k-\eps^2 k}{\eps k}=(1-\eps)\eps^2 k^2$ from the modular function $c$. Since the submodular function $f$ of our hard instance is the sum of the cut function and the modular function $c$, we get $f(OPT)\ge 2(1-\eps)\eps^2 k^2$.
\paragraph{Soundness}
We will prove that for any constant $\delta>0$, w.h.p.~the value of every size-$(\le k)$ set that does not contain any element from $B_g$ is $\le (1+\delta)\eps^2 k^2$. Before that, we show this finishes the proof of soundness: We notice that by construction of our hard instance, $(\le2\eta)$-fraction of $B_g$ can only have marginal value $\le\eps^2k\cdot 2\eta \cdot b\le2\eta\cdot \eps^2k^2$ to any set (specifically, by construction of our hard instance, for any $\tau\in[0,1]$, the value of any $\tau$-fraction of $B_g$ is exactly $\eps^2k\cdot \tau\cdot b$, and by submodularity of our instance, the marginal value of any set $S$ (to any other set) is at most the value of $S$ itself). Therefore, for any constant $\delta>0$, w.h.p.~the value of every size-$(\le k)$ set that does not contain $(\le2\eta)$-fraction of $B_g$ is $\le (1+\delta+2\eta)\eps^2k^2$.

It remains to prove that for any constant $\delta>0$, w.h.p.~the value of every size-$(\le k)$ set $S$ that does not contain any element from $B_g$ is $\le (1+\delta)\eps^2 k^2$. Without loss of generality, we assume $S$ does not contain two vertices from the same $B_i$ for any $i\neq g$ (specifically, this is sufficient because by our construction, for any $i\neq g$, we only have hyperedges from entire $B_i$ to each vertex in $N_i$, and hence including one vertex from $B_i$ is enough to get all the possible value of $B_i$). Furthermore, we assume that $S$ contains a set of $\beta k$ vertices from $A_1$ (denoted by $S_1$) and a set of $\gamma k$ vertices from $A_2$ (denoted by $S_2$) such that $\beta+\gamma\le 1, \gamma\le\eps$ and $\gamma,\beta=\Omega(1)$ (assuming $\gamma,\beta=\Omega(1)$ is without loss of generality: If $\beta=o(1)$, then $S_L$ has value at most $o(\eps^2 k^2)$, and if $\gamma=o(1)$, then $S_R$ has value at most $o(\eps^2 k^2)$, and moreover, both $S_L$ and $S_R$ have value at most $\eps^2k^2$ for any $\beta,\gamma\le 1$, and thus, if $\beta=o(1)$ or $\gamma=o(1)$, what we want to prove trivially holds).

The expectation of the number of hyperedges between $S_L$ and $A_2\setminus S_R$ by our construction is $(\beta k)\cdot (\eps^2 k)\cdot(1-\gamma/\eps)=\beta(\eps-\gamma)\eps k^2$. By Chernoff bound, the number of the hyperedges between $S_L$ and $A_2\setminus S_R$ is $(1\pm \delta)\beta(\eps-\gamma)\eps k^2$ with probability $1-2^{-\Omega(k^2)}$ for any constant $\delta>0$. Moreover, there are at most $\binom{\ell}{1}+\dots+\binom{\ell}{k}\le k\binom{\ell}{k}\le k(\frac{e\ell}{k})^k$ possible choices of $S_L$ (because we assumed at most one element per $B_i$ for all $i\neq g$, and in each $B_i$ such that $i\neq g$, every vertex is essentially the same) and at most $2^{\eps k}$ possible choices of $S_R$. By a union bound and our assumption that $n=2^{o(k)}$, with probability $1-k(\frac{e\ell}{k})^k\cdot 2^{\eps k-\Omega(k^2)} = 1-o(1)$, the number of hyperedges, between every possible $S_L$ of size $\beta k$ and every possible $A_2\setminus S_R$ of size $(\eps-\gamma) k$, is $(1\pm \delta)\beta(\eps-\gamma)\eps k^2$. Therefore, with probability $1-o(1)$, the value of any $S_L\cup S_R$ is $(1\pm \delta)\beta(\eps-\gamma)\eps k^2+(\gamma/\eps)\eps^2k^2$ (the number of hyperedges between $S_L$ and $A_2\setminus S_R$ plus the modular function value $c(S_R)$), which is at most $(1+\delta)(\eps-\gamma)\eps k^2+(\gamma/\eps)\eps^2k^2\le (1+\delta)\eps^2k^2$.
\end{proof}

\bibliographystyle{alpha}
\bibliography{cite}

\appendix
\section{Proof of Proposition~\ref{prop:black_box_reduction}}
We restate Proposition~\ref{prop:black_box_reduction} below.
\begin{proposition}
For symmetric submodular function maximization over a set of $n$ elements $E$ under cardinality constarint $k$, any algorithm guaranteeing a $(1-1/e+\eps)$-approximation must:
\begin{description}
\item[Offline] use $n^{\Omega(k)}$ queries; or
\item[Random-order streaming] use $\Omega(n)$-buffer size.
\end{description}
\end{proposition}
\begin{proof}
The proof is based on the hard instances from~\cite{NW78Bound}. They show that there is a non-negative monotone submodular function $f$ that requires $n^{\Omega(k)}$ queries to achieve $(1-1/e+\eps)$-approximation of the optimum $f(E)$. Moreover, their instances have the property (*) that for any subset of elements $X$ such that $|X|\ge 4k$, $f(X)=f(E)$. The recent work~\cite{LRSVZ21} shows that exactly the same $f$ requires $\Omega(n)$ memory to achieve $(1-1/e+\eps)$-approximation in the random-order streaming model, even with unbounded computational and query complexity.

The hard instance $g$ for proving this proposition is simply defined by $$g(X):=f(X)+f(E\setminus X)-f(E).$$ 
$g$ is non-negative because the instances from~\cite{NW78Bound} are such that $\frac{n}{k}$ can be arbitrarily large, and hence at least one of $X$ and $E\setminus X$ has size $\ge 4k$, then it follows by the property (*) that 
$$g(x) = f(X)+f(E\setminus X)-f(E) = \min\{f(E\setminus X), f(X)\}\ge 0.$$

Moreover, by definition of $g$, we have that $g(E\setminus X)=f(E\setminus X)+f(X)-f(E)=g(X)$, and thus, $g$ is indeed symmetric. Futhermore, $g$ is submodular, because (i) if the function $f(X)$ is submodular, then the function $h(X):=f(E\setminus X)$ is also submodular, and (ii) the sum of submodular functions is still submodular.

Finally, we observe that for any subset of elements $X$ with $|X|\le k$, $g(X)=f(X)$. Indeed, since $|X|\le k$, we have $|E\setminus X|\ge 4k$ (for the instance where $n/k$ is sufficiently large), and hence $f(E\setminus X)=f(E)$ by the property (*), and it follows that $g(X)=f(X)$. Therefore, under the cardinality constraint $k$, maximizing $g$ is equivalent to maximizing $f$. Thus, the lower bounds for maximizing $f$ in the respective settings carry over to maximizing the symmetric function $g$.

\end{proof}

\section{Adapting our hard instance for stronger algorithms}\label{section:arbitrary_output}
In this section, we show how to slightly modify our hard instance in the proof of Theorem~\ref{thm:lower_bound} such that the same hardness result holds against stronger algorithms that are allowed to output any size-$(\le k)$ subset of the ground set $V$, which are non-standard but might be of independent interest to some readers, and we will give a proof sketch based on the proof of Theorem~\ref{thm:lower_bound} (we assume the reader has already read the proof of Theorem~\ref{thm:lower_bound}).

\begin{theorem}
Assuming $n=2^{o(k)}$, any $(1/2+\eps)$-approximation algorithm for non-monotone submodular maximization in the random-order streaming model must use $\Omega(n/k^2)$ memory, even if we allow the algorithm to output any size-$(\le k)$ subset of the ground set at the end of the stream.
\end{theorem}
\begin{proof}[Proof sketch]
A new gadget that we will use to modify the hard instance in the proof of Theorem~\ref{thm:lower_bound} is a directed hyperedge with cap $p\in\mathbb{N}$: A directed hyperedge with cap $p$ in a directed hypergraph is represented by some $(U,v)$, where $U$ is a subset of vertices, and $v\notin U$ is a vertex. For any subset of vertices $S$, a hyperedge $(U,v)$ with cap $p$ is cut by $S$ iff $|U\cap S|>0$ and $v\notin S$, and, moreover, the size of the cut is $\min\{|U\cap S|,p\}$. It is well-known that such cut function is submodular. (For example, the hyperedge we used in the hard instance for Theorem~\ref{thm:lower_bound} is a hyperedge with cap $1$.)

\subsubsection*{Construction of hard instance}
Let $\alpha>0$ be a large integer constant which we will specify later. The construction is same as that in the proof of Theorem~\ref{thm:lower_bound} except that we change step 2 of the random generating procedure to the following:
\begin{enumerate}
    \item[2] Then, we slightly modify the graph generated in step 1 as follows: First, we sample a size-$\alpha$ subset $\mathcal{G}\subseteq[\ell]$, then
    \begin{enumerate}
        \item we sample a uniformly random $g\in\mathcal{G}$. For each $u_j\in N_{g}$, we remove the hyperedge (with cap 1) from $B_{g}$ to $u_j$, and instead, for each $v\in B_{g}$, we create a directed hyperedge (with cap 1) from $\{v\}$ to $u_j$. (That is, for each $u_j\in N_{g}$, we replace the hyperedge (with cap 1) from $B_{g}$ to $u_j$ with individual edges from each $v\in B_{g}$ to $u_j$.)
        \item Moreover, for every $g'\in\mathcal{G}\setminus\{g\}$, we add the following modification: For each $u_j\in N_{g'}$, we remove the hyperedge from $B_{g'}$ to $u_j$, and instead, we create a directed hyperedge with cap $\floor{\frac{b}{\alpha^2}}$ from $B_{g'}$ to $u_j$. (That is, we replace the hyperedge (with cap 1) from $B_{g'}$ to $u_j$ with a hyperedge with cap $\floor{\frac{b}{\alpha^2}}$ from $B_{g'}$ to $u_j$.)
    \end{enumerate}
\end{enumerate}
The final submodular function $f:V\to\RR_{\ge 0}$ is the sum of the cut function on the above generated hypergraph plus the modular function $c(S):=(\eps^2 k^2)\cdot \frac{|S\cap A_2|}{|A_2|}$. See Figure~\ref{fig:hard_instance_2} for an illustration.

\begin{figure}[h]
\centering
    \includegraphics[scale=0.5]{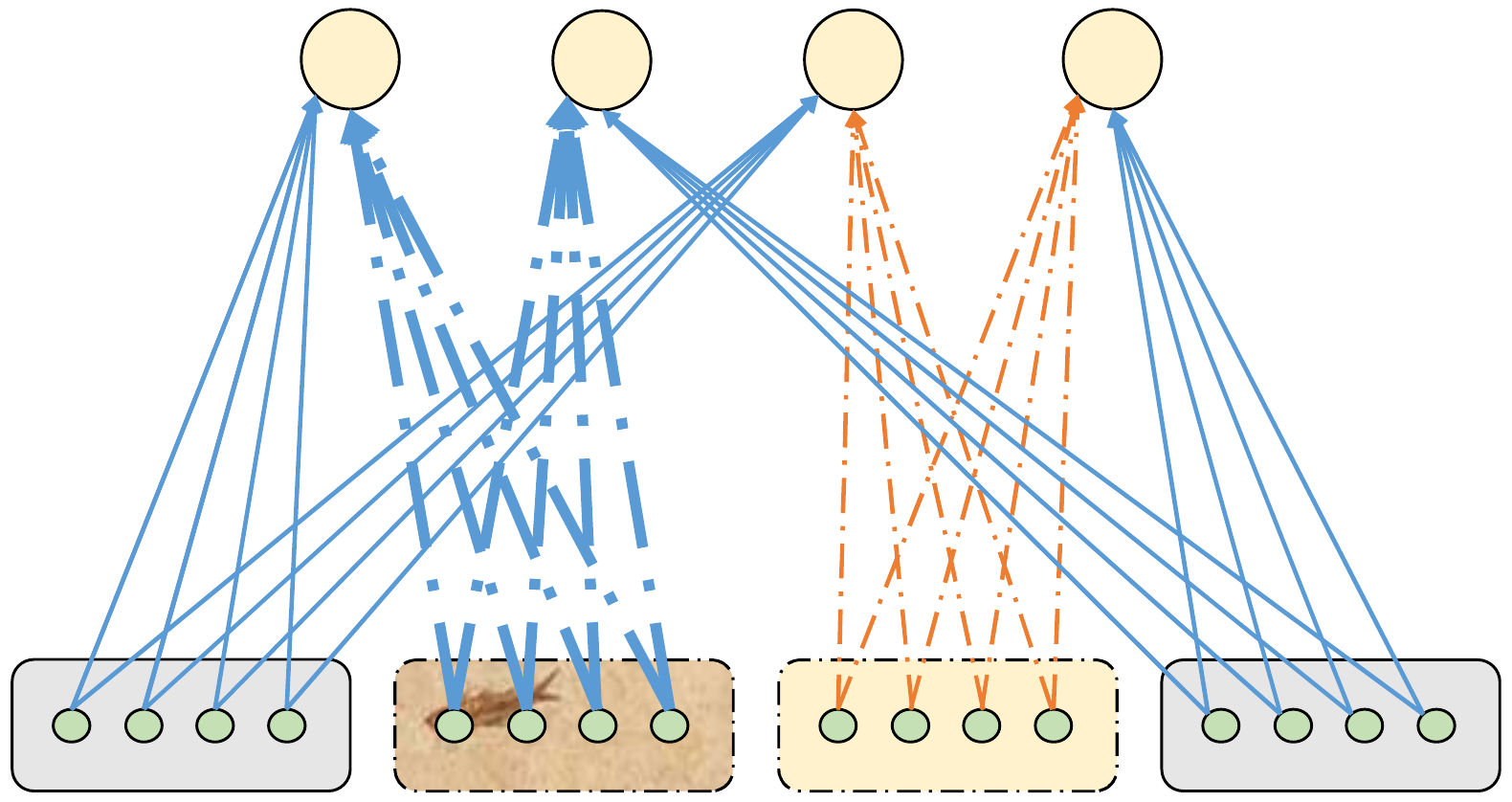}
    \caption{An illustration of our new hard instance: There are $\eps k$ individual edges (orange and dashed) from each vertex in bucket $B_g$ (with yellow filling and dashed outline) to $B_g$'s neighborhood $N_g$, and for each $g'\in\mathcal{G}\setminus\{g\}$, there are $\eps k$ hyperedges (blue and bold dashed) with cap $\floor{\frac{b}{\alpha^2}}$ from bucket $B_{g'}$ (with a mark in the filling and dashed outline) to its neighborhood $N_{g'}$, and there are $\eps k$ hyperedges with cap $1$ (blue and solid) from every other bucket $B_i$ (with gray filling and solid outline) to its neighborhood $N_i$.}
    \label{fig:hard_instance_2}
\end{figure}

Basically, besides the truly good bucket $B_g$ and the bad buckets (which we also had in the hard instance for Theorem~\ref{thm:lower_bound}), we introduce some fake good buckets $B_{g'}$'s to further fool the algorithm in the last $\eta$-fraction of the stream, which we will explain shortly. Now we outline our new proof strategy:
\subsubsection*{Outline of the hardness proof}
\begin{enumerate}
    \item First, we use the random-order property of the stream and union bound to show that given any constant $\eta\in(0,1)$, for any algorithm with $o({n}/{k^2})$ memory, w.h.p.~there are at most $o(n/k)$ collisions occurring in total during the stream before the $(1-\eta)n$-th element of the stream arrives. (This part is exactly same as proof of Theorem~\ref{thm:lower_bound}.)
    \item Then, for any $q<\ell$, we prove by induction that for any $j\in\mathcal{G}$, the probability that one of the first $q$ collisions encountered by the algorithm is in $B_{j}$ is at most $q/(\ell-q)$ (this part is similar to the proof of Theorem~\ref{thm:lower_bound}), which implies that (i) for any $j\in\mathcal{G}$ (and in particular $g$), w.h.p.~none of the first $o({n}/{k})$ collisions is in $B_{j}$. Thus, (i) w.h.p.~none of the first $o({n}/{k})$ collisions is in $B_{g}$, and (ii) in expectation, there are $o(\alpha)$ many $g'\in\mathcal{G}\setminus\{g\}$ for which at least one of the first $o({n}/{k})$ collisions is in $B_{g'}$, and by Markov inequality, w.h.p.~there are only $o(\alpha)$ many $g'\in\mathcal{G}\setminus\{g\}$ for which at least one of the first $o({n}/{k})$ collisions is in $B_{g'}$ (and we let $\mathcal{G}'$ denote the set of the other $\alpha-o(\alpha)$ many $g'\in\mathcal{G}\setminus\{g\}$ for which none of the first $o({n}/{k})$ collisions is in $B_{g'}$). (This part is similar to the proof of Theorem~\ref{thm:lower_bound}.)
    
    Note that this together with the first point imply that for any algorithm, w.h.p.~it has never stored $\ge 2$ elements of $B_g$ or $B_{g'}$ for any $g'\in\mathcal{G}'$ together in its memory before the $(1-\eta)n$-th element of the stream arrives, but to distinguish between a fake good bucket $B_{g'}$ and the truly good bucket $B_g$, the algorithm needs to query a set that contains at least $\floor{\frac{b}{\alpha^2}}+1$ elements of $B_{g'}$ or $B_g$ (because for any other set, by our construction, its value does not depend on which of $B_{g'},B_g$ is the truly good bucket). Hence, w.h.p.~the algorithm has no information about which $j\in \mathcal{G}'\cup\{g\}$ is the index of the truly good bucket before the $(1-\eta)n$-th element of the stream arrives.
    
    The proof outlined so far is similar to what we did in the proof of Theorem~\ref{thm:lower_bound}. Here comes the punchline: We choose $\alpha$ such that $1/\alpha^2\ge 3\eta$. For each $j\in\mathcal{G}$, by standard concentration bound, with probability $1-2^{-\Omega(\eta b)}$, at most $2\eta$-fraction of elements in $B_j$ appears in the last $\eta$-fraction of the random-order stream. By a union bound over $\alpha$ many $j\in\mathcal{G}$, w.h.p.~at most $2\eta$-fraction of $B_j$ appears in the last $\eta$-fraction of the stream for all $j\in\mathcal{G}$, but the algorithm needs to have at least $\floor{\frac{b}{\alpha^2}}+1\ge 3\eta b$ elements to tell whether $B_j$ is the truly good bucket. Combining this with our analysis in the previous paragraph, w.h.p.~the algorithm has no information about which $j\in \mathcal{G}'\cup\{g\}$ is the index of the truly good bucket during the entire stream.
    
    At the end of the stream, by Markov argument, even if we allow the algorithm to output any size-$(\le k)$ subset of $V$, it can only include $\ge k/\sqrt{\alpha}$ elements from $\le \sqrt{\alpha}$ buckets among $B_j$'s for all $j\in G'\cup\{g\}\}$. Moreover, because w.h.p.~the algorithm has no information about which $j\in \mathcal{G}'\cup\{g\}$ is the index of the truly good bucket, it follows that with probability $1-\sqrt{\alpha}/(\alpha-o(\alpha))$, the algorithm only includes $\ge k/\sqrt{\alpha}$ elements from the truly good bucket $B_g$ (note that since $\eta$ is an arbitrarily small constant, and $1/\alpha^2\ge 3\eta$ is the only upper limit for $\alpha$, it follows that $1/\sqrt{\alpha}$ is a constant that can be as small as we want).
    \item Finally, we can show that in our hard instance, there is always a set that contains $B_g$ and has value $2(1-\eps)\eps^2 k^2$ (completeness), but for any constant $\delta>0$, w.h.p.~every size-$(\le k)$ set that only contains $(\le1/\sqrt{\alpha})$-fraction of the elements of $B_g$ has value $\le (1+\delta+1/\sqrt{\alpha})\eps^2k^2$ (soundness), which completes the hardness proof, because $1/\sqrt{\alpha}$ and $\delta,\eps$ are constants that can be arbitrarily small.
    
    Specifically, the completeness proof is exactly same as that for Theorem~\ref{thm:lower_bound}. For the soundness proof, we notice that by our construction, each fake good bucket $B_{g'}$ can contribute marginal value at most $\eps^2k^2/\alpha^2$ to any set (because of submodularity and $f(B_{g'})=\floor{\frac{b}{\alpha^2}}\eps^2k\le\eps^2k^2/\alpha^2$), and hence, all the $\alpha$ fake good buckets can contribute marginal value at most $\eps^2k^2/\alpha$ in total to any set (by submodularity), which is negligible compared to the optimal value because $1/\alpha$ can be arbitrarily small. Thus, we can ignore the possible contribution of all the fake good buckets, and the rest of the soundness proof is same as that for Theorem~\ref{thm:lower_bound}.
\end{enumerate}
\end{proof}

\section{Certificates for the convex programs}
In this section, we provide verifiable proof for the numerical results of our factor-revealing convex programs. (We skip Theorem~\ref{thm:fpt}, since it is an intermediate result that has a weaker approximation guarantee than Theorem~\ref{thm:fpt_plus}.)
\subsection{Theorem~\ref{thm:fpt_plus}}
Recall that in the proof of Theorem~\ref{thm:fpt_plus}, we have the following two convex programs
$$\min_{a,b,c,d,\alpha,\beta_0,\beta_1} r_{\textrm{FPT}}(a,b,c,d,\alpha,\beta_0,\beta_1) \textrm{ s.t. Eq.~\eqref{eq:symmetric_stream_1}~\eqref{eq:symmetric_stream_3}~\eqref{eq:fpt_new_constraint_1}~\eqref{eq:fpt_new_constraint_2}~\eqref{eq:fpt_new_constraint_3}~\eqref{eq:fpt_plus_constraint}},$$
$$\min_{a,b,c,d,\alpha,\beta_0,\beta_1} r_{\textrm{FPT}}(a,b,c,d,\alpha,\beta_0,\beta_1) \textrm{ s.t. Eq.~\eqref{eq:symmetric_stream_1}~\eqref{eq:symmetric_stream_4}~\eqref{eq:fpt_new_constraint_1}~\eqref{eq:fpt_new_constraint_2}~\eqref{eq:fpt_new_constraint_3}~\eqref{eq:fpt_plus_constraint}}.$$
\subsubsection{First convex program}
We write the first convex program explicitly (plugging in $\beta_1=0$ by Eq.~\eqref{eq:fpt_plus_constraint}).
\begin{align*}
    &\min_{a,b,c,d,\alpha,\beta_0}\max\{b,\,a+b,\,c,\,\alpha+c\} \\
    \textrm{s.t. }
    &a+2b\ge 1\\
    &(1-c-b)^2+4(1-a-b)(1-a-2b)\le 0\\
    &2\alpha\ge\beta_0 \\
    &\beta_0+c\ge 1 \\
    &1-c-b\ge 0\\
    &a+b\ge c\\
    &a+b,\,c+d\le 0.9 \\
    &a\le c \\
    & 0\le b,c,d\le 0.9\\
    &0\le \alpha,\beta_0\le 1.
\end{align*}
Observe that $d$ is essentially a free variable that can be set equal to $0$. Moreover, given any feasible solution, we can let $\beta_0=\beta_0^*:=1-c$ instead, which does not change the objective value, and the solution is still feasible (indeed, $0\le\beta_0^*\le 1$ follows from $0\le c\le 0.9$, and $\beta_0^*+c\ge 1$ obviously holds, and finally, $2\alpha\ge \beta_0^*=1-c$ follows from constraints $2\alpha\ge\beta_0$ and $\beta_0+c\ge 1$). Furthermore, since $\alpha+c$ is always better than $c$ by constraint $\alpha\ge0$, we can remove the term $c$ in the maximum operator. Therefore, the program can be simplified to
\begin{align}
    &\min_{a,b,c,\alpha}\max\{b,\,a+b,\,\alpha+c\} \\
    \textrm{s.t. }
    &(1-c-b)^2+4(1-a-b)(1-a-2b)\le 0\nonumber\\
    &2\alpha\ge1-c \nonumber\\
    &1-c-b\ge 0\nonumber\\
    &a+2b\ge 1\nonumber\\
    &a+b\ge c\nonumber\\
    &a+b\le 0.9 \nonumber\\
    &a\le c \nonumber\\
    & 0\le b,c\le 0.9\nonumber\\
    &0\le \alpha\le 1\nonumber.
\end{align}
We prove that if there is a feasible solution with objective value $\le 0.539$, then the constraint $(1-c-b)^2+4(1-a-b)(1-a-2b)\le 0$ must be violated, which is a contradiction. Specifically, if the objective value is $\le 0.539$, then $b,a+b,\alpha+c\le 0.539$. Notice that because of the term $\alpha+c$ in the maximum operator, the minimizer wants $\alpha$ to be as small as possible, and hence the constraint $2\alpha\ge 1-c$ should always be tight (and the constraint $0\le\alpha\le 1$ can be ignored because $0.1\le 1-c\le 1$). Then, it follows from $2\alpha=1-c$ and $\alpha+c\le 0.539$ that $c\le 2\times0.539-1$, and thus, by $b\le 0.539$ and $c\le 2\times0.539-1$, we have 
\begin{equation}\label{eq:1_of_A_1_1}
    1-c-b\ge 1-(2\times0.539-1)-0.539=2-3\times0.539.
\end{equation}

On the other hand, $(1-a-b)(1-a-2b)=(b+1-a-2b)(1-a-2b)\ge (0.539+1-a-2b)(1-a-2b)$ by $1-a-2b\le0$ and $b\le 0.539$. If we think of $a+2b$ as a variable, $(0.539+1-a-2b)(1-a-2b)$ is a quadratic function, which decreases as $a+2b$ increases when $a+2b<\frac{2+0.539}{2}$. By $a+b\le 0.539$ and $b\le 0.539$, we have that $a+2b\le 2\times0.539<\frac{2+0.539}{2}$, and therefore, $(0.539+1-a-2b)(1-a-2b)$ is minimized at $a+2b=2\times0.539$. It follows that 
\begin{equation}\label{eq:2_of_A_1_1}
    (1-a-b)(1-a-2b)\ge (1-0.539)(1-2\times 0.539).
\end{equation}

Combining Eq.~\eqref{eq:1_of_A_1_1} and~\eqref{eq:2_of_A_1_1}, we get $(1-c-b)^2+4(1-a-b)(1-a-2b)>(2-3\times0.539)^2+4(1-0.539)(1-2\times 0.539)>0$, which is the contradiction. We still need to argue the program has a feasible solution. One can numerically check $(a=0,\,b=\frac{2(6-\sqrt{2})}{17},\,c=-1+\frac{4(6-\sqrt{2})}{17},\,\alpha=1-\frac{2(6-\sqrt{2})}{17})$ is feasible.

\subsubsection{Second convex program}
We write the second convex program explicitly (similar to the first convex program, we plug in $\beta_1=0$ and $\beta_0=1-c$, and remove the term $c$ in the maximum operator of the objective).
\begin{align}
    &\min_{a,b,c,\alpha}\max\{b,\,a+b,\,\alpha+c\} \\
    \textrm{s.t. }
    &1-a-2b\le 0\nonumber\\
    &2\alpha\ge1-c \nonumber\\
    &1-c-b\le 0\nonumber\\
    &a+2b\ge 1\nonumber\\
    &a+b\ge c\nonumber\\
    &a+b\le 0.9 \nonumber\\
    &a\le c \nonumber\\
    & 0\le b,c\le 0.9\nonumber\\
    &0\le \alpha\le 1\nonumber.
\end{align}
The first and the fourth constraints imply that $a+2b=1$, and hence, $a+b=1-b$. Similar to the first convex program, the minimizer should satisfy $2\alpha=1-c$. Hence, $\alpha+c=\frac{1+c}{2}$, and by the constraint $1-c-b\le0$, we have that $\frac{1+c}{2}\ge 1-\frac{b}{2}$.

Together, we get $\frac{2}{5}b+\frac{1}{5}(a+b)+\frac{2}{5}(\alpha+c)\ge\frac{2}{5}b+\frac{1}{5}(1-b)+\frac{2}{5}(1-\frac{b}{2})=0.6$, which implies that the minimum of the program has to be at least $0.6$. Finally, the program has feasible solutions, e.g., $(a=-\frac{1}{3},\,b=\frac{2}{3},\,c=\frac{1}{3},\,\alpha=\frac{1}{3})$.

\subsection{Theorem~\ref{thm:symmetric_beating_half}}
Recall that in the proof of Theorem~\ref{thm:fpt_plus}, we have the following two convex programs
$$\min_{a,b,c,d} r(a,b,c,d) \textrm{ s.t. Eq.~\eqref{eq:symmetric_stream_1} and~\eqref{eq:symmetric_stream_3}},$$
$$\min_{a,b,c,d} r(a,b,c,d) \textrm{ s.t. Eq.~\eqref{eq:symmetric_stream_1} and~\eqref{eq:symmetric_stream_4}}.$$
\subsubsection{First convex program}
\begin{align}
    &\min_{a,b,c,d}\max\{b,\,a+b,\,d+c,\,(2-1/e)(1-a-b-d)-d-a\} \\
    \textrm{s.t. }
    &(1-c-b)^2+4(1-a-b)(1-a-2b)\le 0\nonumber\\
    &1-c-b\ge 0\nonumber\\
    &a+2b\ge 1\nonumber\\
    &a+b\ge c\nonumber\\
    &a+b,\,c+d\le 0.9 \nonumber\\
    &a\le c \nonumber\\
    & 0\le b,c,d\le 0.9\nonumber.
\end{align}
We show that if a feasible solution of the program has objective value $\le 0.5029$, then the constraint $(1-c-b)^2+4(1-a-b)(1-a-2b)\le 0$ must be violated, which is a contradiction. Specifically, for such feasible solution, we know $b\le 0.5029$, $a+b\le 0.5029$, $d+c\le 0.5029$, and $(2-1/e)(1-a-b-d)-d-a\le 0.5029$. By $a+b\le 0.5029$ and the constraint $a+2b\ge 1$, we get $a\le 0.0058$.

Starting from $(2-1/e)(1-a-b-d)-d-a\le 0.5029$, we derive that
\begin{align*}
    0.5029&\ge(2-1/e)(1-a-b-d)-d-a \\
    &\ge(2-1/e)(1-0.5029-d)-d-a &&\text{($a+b\le 0.5029$)} \\
    &\ge(2-1/e)(1-0.5029-d)-d-0.0058 &&\text{($a\le 0.0058$)}.
\end{align*}
By rearranging, we get $d\ge \frac{(2-1/e)(1-0.5029)-0.5087}{3-1/e}$. Since $d+c\le 0.5029$, we have that $c\le 0.5029-\frac{(2-1/e)(1-0.5029)-0.5087}{3-1/e}$. Combining with $b\le 0.5029$, we get
\begin{equation}\label{eq:1_of_A_2_1}
1-c-b\ge -0.0058+\frac{(2-1/e)(1-0.5029)-0.5087}{3-1/e}.    
\end{equation}

On the other hand, $(1-a-b)(1-a-2b)=(b+1-a-2b)(1-a-2b)\ge (0.5029+1-a-2b)(1-a-2b)$ by $1-a-2b\le0$ and $b\le 0.5029$. If we think of $a+2b$ as a variable, $(0.5029+1-a-2b)(1-a-2b)$ is a quadratic function, which decreases as $a+2b$ increases when $a+2b<\frac{2+0.5029}{2}$. By $a+b\le 0.5029$ and $b\le 0.5029$, we have that $a+2b\le 2\times0.5029<\frac{2+0.5029}{2}$, and therefore, $(0.5029+1-a-2b)(1-a-2b)$ is minimized at $a+2b=2\times0.5029$. It follows that 
\begin{equation}\label{eq:2_of_A_2_1}
    (1-a-b)(1-a-2b)\ge (1-0.5029)(1-2\times 0.5029).
\end{equation}

Combining Eq.~\eqref{eq:1_of_A_2_1} and~\eqref{eq:2_of_A_2_1}, we get $(1-c-b)^2+4(1-a-b)(1-a-2b)>(-0.0058+\frac{(2-1/e)(1-0.5029)-0.5087}{3-1/e})^2+4(1-0.5029)(1-2\times 0.5029)>0$, which is the contradiction. Finally, one can numerically check $(a=10^{-8},\,b=0.5034476995316219,\,c=0.3795235846990063,\,d=0.12392128506019447)$ is feasible.

\subsubsection{Second convex program}
\begin{align}
    &\min_{a,b,c,d}\max\{b,\,a+b,\,d+c,\,(2-1/e)(1-a-b-d)-d-a\} \\
    \textrm{s.t. }
    &1-a-2b\ge 0\nonumber\\
    &1-c-b\le 0\nonumber\\
    &a+2b\ge 1\nonumber\\
    &a+b\ge c\nonumber\\
    &a+b,\,c+d\le 0.9 \nonumber\\
    &a\le c \nonumber\\
    & 0\le b,c,d\le 0.9\nonumber.
\end{align}
By the constraints, we have that $a+2b=1$. If the minimum of the program is $\le 0.5029$, then $a+b\le 0.5029$, $b\le 0.5029$, $d+c\le 0.5029$, and $(2-1/e)(1-a-b-d)-d-a\le 0.5029$. By $a+b\le 0.5029$ and $a+2b=1$, we have $a\le 0.0058$. By the constraint $1-c-b\le 0$ and $b\le 0.5029$, we get $c\ge 1-0.5029$, and since $d+c\le 0.5029$, we have $d\le 0.0058$. Now we can derive
\begin{align*}
    (2-1/e)(1-a-b-d)-d-a&\ge (2-1/e)(1-0.5029-d)-d-a &&\text{($a+b\le 0.5029$)}\\
    &\ge (2-1/e)(1-0.5029-d)-d-0.0058 &&\text{($a\le 0.0058$)}\\
    &\ge (2-1/e)(1-0.5029-0.0058)-0.0058-0.0058 &&\text{($d\le 0.0058$)}\\
    &>0.79,
\end{align*}
which contradicts $(2-1/e)(1-a-b-d)-d-a\le 0.5029$. Finally, to see the program is feasible, one can numerically check feasibility of $(a=0,\,b=\frac{5e-2}{3(3e-1)},\,c=\frac{4e-1}{3(3e-1)},\,d=\frac{e-1}{3(3e-1)})$.

\end{document}

%% file: intro.tex
\section{Introduction}
We study the algorithmic problem of selecting a small subset of $k$ elements out of a (very) large ground set of $n$ elements.
In particular, we want the small subset to consist of $k$ elements that are valuable {\em together}, as is captured by an objective function $f:2^{E}\to\mathbb{R}_{\ge 0}$. Without any assumptions on $f$ it is hopeless to get efficient algorithms; we make the assumption that $f$ is {\em submodular}%
\footnote{I.e., functions where the marginal value of an element is decreasing as the set grows.}, one of the most fundamental  and well-studied assumptions in combinatorial optimization.

Alas, even for submodular functions, strong impossibility results are known. Two of the most important frameworks we have for circumventing  impossibility results are 
(i) approximation algorithms --- look for solutions that are only approximately optimal; and (ii) paramerterized complexity --- look for algorithms of which the runtime is efficient as a function of the large ground set $n$, but may have a worse dependence%
\footnote{Formally, an algorithm is said to be fixed-parameter tractable if it runs in time $h(k)\cdot\poly(n)$ for {\em any} function $h$. Here $h$ could be arbitrarily fast growing, e.g.~doubly-exponential or Ackermann --- this is asymptotically faster than the naive $n^k$. In this work we will be more ambitious (and closer to practice) and present algorithms that run in time
$2^{\tilde{O}(k)}\cdot n$.} on the smaller parameter $k$. To appreciate the relevance of parameterized complexity in practice, it's important to note that in many applications of submodular maximization $k$ is indeed quite small, e.g.~in data summarization~\cite{BMKK14}, we want an algorithm that, given a large image dataset  chooses a representative subset of images that must be small enough to fit our screen.

Submodular optimization has been thoroughly studied under the lens of approximation algorithms; much less work has been done about its parameterized complexity (with the exception of~\cite{Skowron17}, see discussion of related works). In either case, strong, tight hardness results are known. In this work we show that combining both approaches gives surprisingly powerful algorithms.

On the technical level, we develop novel insights to design and analyze fixed-parameter tractable (FPT) algorithms for (non-monotone) submodular function maximization. We instantiate these ideas to give new results in two settings: offline (``classical'') algorithms for which we are interested in the running time and query complexity\footnote{I.e., the function $f$ is given as a value oracle. The query complexity is number of queries made by the algorithm, which is clearly a lower bound of the runtime.}, and random-order streaming algorithms for which we mostly care about the memory cost. 

\subsection*{Main result I: offline algorithms}
Our first result is an (offline) FPT algorithm that guarantees an improved approximation ratio for submodular function maximization.

To compare our result with the approximation factors achievable by polynomial-time algorithms, we first briefly survey the existing algorithmic and hardness results.

On the algorithmic side, the current state-of-art polynomial-time algorithm for general non-monotone submodular functions achieves $0.385$-approximation~\cite{BF19}. For sub-classes of submodular functions, better polynomial-time approximation algorithms are known: notable examples include monotone functions (the greedy algorithm achieves $(1-1/e)$-approximation~\cite{NWF78}), and symmetric functions%
\footnote{I.e.~functions that assign the same value to a set and its complement, which capture some of the most important applications of non-monotone submodular functions, including mutual information and cuts in (undirected) graphs and hypergraphs.} (the state-of-the-art approximation factor is $0.432$~\cite{Feldman17}).

From the hardness perspective, known results rule out {\em polynomial} query complexity algorithms with approximation factors better than $1/2$ or $0.491$ for symmetric~\cite{FMV11,V13} or asymmetric functions~\cite{GV11}, respectively. It is also known that even FPT algorithms cannot beat  $(1-1/e)$-approximation, and this holds even  for monotone submodular functions~\cite{NW78Bound}.

In FPT time, the streaming algorithm of~\cite{AEFNS20} implies a $1/2$-approximation algorithm for general non-monotone functions (although it is not explicitly stated as an FPT algorithm in their paper), which slightly beats the aforementioned $0.491$ bound for asymmetric functions. However, this result does not tell us whether FPT algorithms can beat the $1/2$ bound for symmetric functions. A-priori, it was plausible that $1/2$-approximation is the best achievable approximation by FPT algorithms for symmetric functions and hence general non-monotone functions. (In fact, prior to discovering our new algorithms, we had expected that the $1/2$-approximation would indeed be the best that FPT algorithms can achieve.)

We are thus excited to report that we were able to design an FPT algorithm (Algorithm~\ref{alg:fpt_plus}) that not only outperforms all the previous algorithms but also surpasses all the upper limits on the approximation factor in the existing hardness results, {\em by a significant margin, regardless of whether the function is symmetric or not}:
\begin{theorem}[FPT algorithm]
There is a $0.539$-approximation algorithm for cardinality constrained submodular maximization that has runtime and query complexity $2^{\widetilde{O}(k)}\cdot n$.
\end{theorem}

\subsubsection*{Main result II: random-order streaming} 
Our FPT algorithm (Algorithm~\ref{alg:fpt_plus}) uses a subroutine (Algorithm~\ref{alg:symmetric}) which can be interpreted as a {\em random-order streaming algorithm}, and thus, in addition to our FPT result, we also hope to understand the power and the limit of Algorithm~\ref{alg:symmetric} in the random-order streaming model. In this model (see the detailed setup in Section~\ref{section:prelim}), a streaming algorithm makes a single pass over a stream of elements arriving in a uniformly random order. An algorithm keeps a carefully chosen subset of the elements it has seen in a buffer of bounded size. At any point in the stream, the algorithm can make unlimited queries of the function values on subsets of the elements in the buffer. The goal is to obtain a good approximation of the offline optimum while keeping the buffer small (ideally, polynomial in $k$ and independent of $n$).

We show that Algorithm~\ref{alg:symmetric} achieves $1/2$-approximation for general non-monotone submodular functions\footnote{The $1/2$-approximation of Algorithm~\ref{alg:symmetric} for general non-monotone functions is not interesting by itself -- $1/2$-approximation was achieved even if the elements arrive in the worst-case order~\cite{AEFNS20}. However, Algorithm~\ref{alg:symmetric}, which takes advantage of the random order, led us to the discovery of the $1/2$-hardness in the random-order setting.} and beats $1/2$-approximation for symmetric functions using $\widetilde{O}(k^2)$-size buffer\footnote{These algorithmic results also hold for the (similar but incomparable) secretary with shortlists model~\cite{SMC19}.}: 

\begin{theorem}[Random-order streaming algorithm]
For cardinality constrained submodular function maximization in the random-order streaming model, there is an algorithm using $\widetilde{O}(k^2)$-size buffer that achieves $1/2$-approximation for general non-monotone submodular functions and $0.5029$-approximation for symmetric submodular functions.
\end{theorem}

We complement the algorithmic result with a tight $1/2$-hardness result in the random-order streaming setting:
\begin{theorem}[$1/2$-hardness for random-order streaming]
If $n=2^{o(k)}$, any $(1/2+\eps)$-approximation algorithm for cardinality-constrained non-monotone submodular maximization in the random-order streaming model must use an $\Omega(n/k^2)$-size buffer. In fact, this hardness result holds against stronger algorithms that are not captured by the standard random-order streaming model for submodular maximization (see Remark~\ref{rmk:stronger_algorithms}).
\end{theorem}
This hardness result is quite surprising because it shows in contrast to monotone submodular maximization, non-monotone submodular maximization in the random-order setting is {\em not any easier than} that in the worst-order setting where the elements arrive in the worst-case order -- for worst-order streaming model, it is known that $\Omega(n/k^3)$-size buffer is required to beat $1/2$-approximation even if the submodular function is monotone~\cite{FNSZ20}, but for random-order model, recent work by~\cite{SMC19} gives a $(1-1/e-\eps)$-approximate algorithm using $O(k/2^{\textrm{poly}(\eps)})$-size buffer, which is improved to $O(k/\eps)$ by a simpler algorithm of~\cite{LRSVZ21}.

Furthermore, notice that our algorithmic result and hardness result together exhibit a separation between symmetric and asymmetric submodular functions in the random-order streaming setting. This separation is interesting because in the literature, tight hardness result for general non-monotone functions often continues to hold for symmetric functions. For example, for unconstrained non-monotone submodular maximization, there is a family of symmetric functions for which $(1/2+\eps)$-approximation requires $2^{\Omega(n)}$ queries~\cite{FMV11,V13}, and there are efficient matching $1/2$-approximation algorithms even for asymmetric functions~\cite{BFNS15}.

To the best of our knowledge, our result is the {\em first provable separation of symmetric and asymmetric submodular maximization in any setting}, let alone a natural setting that gains a lot of interests recently. Although admittedly we would be more excited to see such separation in the more classic offline setting of constrained non-monotone submodular maximization, currently we are still far from figuring out whether there is a separation in this setting (because of the gap between the current best $0.432$-approximation algorithm for symmetric functions and the current best $0.491$-hardness for general asymmetric functions we mentioned earlier), and we hope our result can provide some insights on how to resolve this problem. 

\subsubsection*{Future directions}
Our FPT algorithm achieves significantly better-than-$1/2$ approximation for general non-monotone submodular functions in the offline setting, and its subrountine Algorithm~\ref{alg:symmetric} breaks the $1/2$ hardness for symmetric submodular functions in the random-order streaming setting -- but what is the best possible approximation ratio in those respective settings? We leave this as an open problem for future work.

We remark that no FPT (small-buffer resp.) algorithm can break the $1-1/e$ barrier for symmetric functions in the offline (random-order streaming resp.) setting. For asymmetric functions, this follows from the classic work of~\cite{NW78Bound} for monotone submodular functions which we mentioned earlier. For symmetric submodular maximization, the monotone functions exhibiting $(1-1/e)$-hardness are obviously not symmetric; but in the appendix, we are able to give a simple black-box approximation-preserving reduction from symmetric non-monotone to asymmetric monotone submodular function maximization that works in both the offline and random-order streaming settings:

\begin{proposition}\label{prop:black_box_reduction}
For cardinality-constrained symmetric submodular function maximization, any algorithm guaranteeing a $(1-1/e+\eps)$-approximation must:
\begin{description}
\item[Offline] use $n^{\Omega(k)}$ queries; or
\item[Random-order streaming] use $\Omega(n)$-buffer size.
\end{description}
\end{proposition}

\subsection{Additional related work}
\paragraph{FPT submodular optimization}
The study of parameterized complexity of submodular maximization was initiated by~\cite{Skowron17} who focused on monotone submodular functions. \cite{Skowron17} gives an FPT approximation scheme for monotone submodular functions that are either $p$-separable or have a bounded ratio of total singleton contribution ($\sum_{e \in E} f(\{e\})$) to total value ($f(E)$). However, for general monotone submodular functions even FPT algorithms (in terms of query complexity) cannot break the classic $1-1/e$ barrier~\cite{NW78Bound}. Furthermore, even for the special case of max-$k$-cover, no FPT algorithms can beat $1-1/e$ assuming gap-ETH~\cite{CGKLL19,Manurangsi20}.

\paragraph{Streaming submodular optimization}
Our work is related to recent works on maximizing submodular functions in random order streams~\cite{SMC19, LRSVZ21, S20b}, but all the latter focus on monotone functions.
Submodular optimization in worst-order streaming models has also been extensively studied in recent years, e.g.~\cite{BMKK14, CGQ15, KMVV15,EDFK17, FKK18, MJK18, AF19, IV19,KMZLK19, MV19,  AEFNS20, HKFK20, HKMY20, FNSZ20}. In the worst-order literature, most relevant to our work is~\cite{AEFNS20} who gave a $1/2$-approximation for general (non-monotone) submodular functions, and~\cite{FNSZ20} who proved a matching inapproximability.